\documentclass{article}

\usepackage[lined,boxed,ruled,vlined,linesnumbered,nofillcomment]{algorithm2e}
\usepackage{bm,bbm}
\usepackage{tabularx,ragged2e,booktabs,caption,framed}
\usepackage[utf8]{inputenc}
\usepackage{float,url,hyperref,epic,eepic}
\usepackage{xcolor, enumerate}
\usepackage{graphicx}
\usepackage{amsmath, amsfonts, amssymb, amscd} 
\usepackage{calc}
\usepackage[numbers,square]{natbib}

\usepackage{listings}
\makeatletter
\newcommand{\pushright}[1]{\ifmeasuring@#1\else\omit\hfill$\displaystyle#1$\fi\ignorespaces}
\newcommand{\pushleft}[1]{\ifmeasuring@#1\else\omit$\displaystyle#1$\hfill\fi\ignorespaces}
\makeatother

\DeclareMathOperator{\poly}{poly}

\DeclareMathOperator{\diag}{diag}
\DeclareMathOperator{\Ring}{R}
\DeclareMathOperator{\cop}{cpr}
\DeclareMathOperator{\sgn}{sgn}

\DeclareMathOperator{\prim}{prim}
\DeclareMathOperator{\adj}{adj}
\DeclareMathOperator{\SYM}{sym}
\DeclareMathOperator{\ord}{ord}

\DeclareMathOperator{\tsym}{\textsc{sym}} 

\DeclareMathOperator{\sym}{\gamma}

\newcommand{\sympk}{\tsym_{p^k}}
\newcommand{\symtk}{\tsym_{2^k}}

\newcommand{\copt}{\cop_2}
\newcommand{\copp}{\cop_p}

\DeclareMathOperator{\fspkt}{\fS^t_{p^k}}
\DeclareMathOperator{\fspksym}{\fS^{\sym}_{p^k}}
\DeclareMathOperator{\fypk}{\fY_{p^k}}

\DeclareMathOperator{\ypk}{Y_{p^k}}
\DeclareMathOperator{\ytk}{Y_{2^k}}
\DeclareMathOperator{\spk}{S_{p^k}}

\DeclareMathOperator{\spkt}{S^t_{p^k}}

\DeclareMathOperator{\fapk}{\fA_{p^k}}
\DeclareMathOperator{\fbpk}{\fB_{p^k}}
\DeclareMathOperator{\fcpk}{\fC_{p^k}}

\DeclareMathOperator{\sapk}{A_{p^k}}
\DeclareMathOperator{\sbpk}{B_{p^k}}
\DeclareMathOperator{\scpk}{C_{p^k}}

\newcommand{\zRz}[1]{\bbZ/#1\bbZ}
\newcommand{\zqz}{\zRz{q}}
\newcommand{\zpz}{\zRz{p}}
\newcommand{\ztz}{\zRz{2}}
\newcommand{\ztkz}{\zRz{2^k}}

\newcommand{\zpkz}{\zRz{p^k}}

\newcommand{\modpk}{\bmod{p^k}}
\newcommand{\modtk}{\bmod{2^k}}
\newcommand{\eqq}{\overset{q}{\sim}}

\newcommand{\eqp}{\overset{p^*}{\sim}}
\newcommand{\eqpk}{\overset{p^k}{\sim}}

\newcommand{\ordp}{\ord_p}
\newcommand{\ordt}{\ord_2}
\newcommand{\sgnp}{\sgn_p}
\newcommand{\sgnt}{\sgn_2}

\newcommand{\MA}{\mathtt{A}} \newcommand{\MB}{\mathtt{B}}
 \newcommand{\MD}{\mathtt{D}}
 
 \newcommand{\MI}{\mathtt{I}}

\newcommand{\MQ}{\mathtt{Q}} \newcommand{\MS}{\mathtt{S}}
 \newcommand{\MU}{\mathtt{U}}
\newcommand{\MV}{\mathtt{V}}

\newcommand{\Vb}{\mathbf{b}}

\newcommand{\Vx}{\mathbf{x}} 
\newcommand{\Vy}{\mathbf{y}}

\newcommand{\Vv}{\mathbf{v}}
\newcommand{\Vw}{\mathbf{w}}

\newcommand{\typet}{type II$^*$}

\DeclareMathOperator{\ORD}{\textsc{ord}}
\DeclareMathOperator{\SGN}{\textsc{sgn}}

\DeclareMathOperator{\gl}{GL}
\DeclareMathOperator{\SL}{SL}
\DeclareMathOperator{\gln}{GL_n}

\DeclareMathOperator{\sln}{SL_n}

\newcommand{\legendre}[2]{%
\left( \frac{#1}{#2} \right)%
}

\DeclareMathOperator{\gen}{Gen}

\newtheorem{fact}{\textsc{Fact}}

\newcommand{\Vxtilde}{\tilde{\Vx}}

\newcommand{\bbR}{\mathbb{R}}

\newcommand{\bbZ}{\mathbb{Z}}
\newcommand{\bbP}{\mathbb{P}}

\newcommand{\fA}{\mathcal{A}}\newcommand{\fB}{\mathcal{B}}
\newcommand{\fC}{\mathcal{C}}

\newcommand{\fS}{\mathcal{S}}

\newcommand{\fY}{\mathcal{Y}}

\iftrue
\renewcommand{\fA}{\mathcal{A}}\renewcommand{\fB}{\mathcal{B}}
\renewcommand{\fC}{\mathcal{C}}

\renewcommand{\fS}{\mathcal{S}}

\renewcommand{\fY}{\mathcal{Y}}
\fi

\DeclareSymbolFont{bbold}{U}{bbold}{m}{n}
\DeclareSymbolFontAlphabet{\mathbbold}{bbold}

\newcommand{\AND}{\textbf{ and }}

\newsavebox{\theorembox} \newsavebox{\lemmabox}
\newsavebox{\corollarybox} \newsavebox{\propositionbox}
\newsavebox{\examplebox} \newsavebox{\conjecturebox}
\newsavebox{\algbox} \newsavebox{\qbox} \newsavebox{\problembox}
\newsavebox{\definitionbox} \newsavebox{\assumptionbox}
\newsavebox{\hypothesisbox} \newsavebox{\obsbox} 
\savebox{\theorembox}{\noindent\bf Theorem} 
\savebox{\lemmabox}{\noindent\bf Lemma}
\savebox{\corollarybox}{\noindent\bf Corollary}
\savebox{\propositionbox}{\noindent\bf Proposition}
\savebox{\examplebox}{\noindent\bf Example}
\savebox{\conjecturebox}{\noindent\bf Conjecture}
\savebox{\algbox}{\noindent\bf Algorithm}
\savebox{\qbox}{\noindent\bf Question}
\savebox{\definitionbox}{\noindent\bf Definition}
\savebox{\problembox}{\noindent\bf Problem}
\savebox{\assumptionbox}{\noindent\bf Assumption}
\savebox{\hypothesisbox}{\noindent\bf Hypothesis}
\savebox{\obsbox}{\noindent\bf Observation}
\newtheorem{theorem}{\usebox{\theorembox}}
\newtheorem{lemma}[theorem]{\usebox{\lemmabox}}

\newtheorem{definition}{\usebox{\definitionbox}}
\newcommand{\qed}{\;\;\;\Box} 
\newenvironment{proof}{{\bf Proof:}}{\hfill\(\qed\)\newline}

\title{Sampling a Uniform Random Solution of a Quadratic Equation Modulo $p^k$}
\author{
	Chandan Dubey\\
	\texttt{chandan.dubey@inf.ethz.ch}
	\and
	Thomas Holenstein\\
	\texttt{thomas.holenstein@inf.ethz.ch}
}
\date{Institut f\"ur Theoretische Informatik, ETH Z\"urich}

\begin{document}

\maketitle

\begin{abstract}
An $n$-ary integral quadratic form is a formal expression 
$Q(x_1,\cdots,x_n)=\sum_{1\leq i,j\leq n}a_{ij}x_ix_j$ in $n$-variables 
$x_1,\cdots,x_n$, where $a_{ij}=a_{ji} \in \bbZ$.
We present a $\poly(n,k, \log p, \log t)$ randomized algorithm
that given a quadratic form $Q(x_1,\cdots,x_n)$, a prime $p$, a 
positive integer $k$ and an integer $t$,
samples a uniform solution of $Q(x_1,\cdots,x_n)\equiv t \bmod{p^k}$.
\end{abstract}

\section{Introduction}

Let $\Ring$ be a commutative ring with unity and $\Ring^\times$ 
be the set of units (i.e., invertible elements) of $\Ring$. A 
quadratic form over the ring $\Ring$ in $n$-formal variables 
$x_1,\cdots,x_n$ in an expression 
$\sum_{1\leq i, j \leq n}a_{ij}x_ix_j$, where $a_{ij}=a_{ji} 
\in \Ring$. A quadratic form can equivalently be represented by
a symmetric matrix $\MQ^n=(a_{ij})$ such that 
$Q(x_1,\cdots,x_n)=(x_1,\cdots,x_n)'\MQ(x_1,\cdots,x_n)$. The
quadratic form is called integral if $\Ring=\bbZ$ and the 
determinant of the quadratic form $Q$ is defined as $\det(\MQ)$.

Quadratic forms are central to various branches of Mathematics, including
number theory, linear algebra, group theory, and Lie theory.
They also appear in several areas of Computer Science like Cryptography and
Lattices. Several modern factorization algorithms, including Dixon's 
algorithm \cite{Dixon81}, the continued fractions method, and the quadratic
sieve; try to solve $x^2\equiv t\bmod{n}$, where $n$ is the number being
factorized. They also arise naturally as 
the $\ell_2$ norm of lattice vectors. 

It is not surprising that the study of quadratic forms predates Gauss, 
who gave the law of quadratic reciprocity and contributed a great deal
in the study of quadratic forms, including a complete classification of
binary quadratic forms (i.e., $n=2$). Another giant leap was made by
Minkowski in his ``Geometry of Numbers'' \cite{Minkowski10}, which
proposed a geometric method to solve problems in number theory. 
Minkowski also gave explicit formulae to calculate the number of 
solutions $\Vx=(x_1,\cdots,x_n) \in (\zpkz)^n$ to the equation 
$\Vx'\MQ\Vx \equiv t \modpk$ albeit they contain errors for the 
$p=2$ case, see \cite{Watson76}.


The general problem is to solve quadratic equations of the form
$\Vx'\MQ\Vx\equiv t \bmod{q}$ in several variables, where $q, t$
are integers. 
If the factorization of $q$ is known then using
Chinese Remainder Theorem one can show that it is enough to solve
the equation when $q$ is a prime power i.e., $q=p^k$ for some prime
$p$ and positive integer $k$.

We remark that typically 
mathematicians are mainly interested in counting the number
of solutions if $k$ is ``large enough''.
One reason for this is that once $k$ is large enough, increasing $k$ by $1$
simply multiplies the number of solutions by $p^{n-1}$.
Another reason is that the corresponding normalized quantity (the local
density, which is the number of solutions divided by $p^{k(n-1)}$
for $k$ large enough) 
seems to be the ``mathematically natural quantity''.
It arises in many places, for example 
in (some forms of) the celebrated Siegel mass formula \cite{Siegel35}.

Several alternatives are available for counting; mainly aimed
towards computing the local density, which is a more general
problem
\cite{Siegel35,OMeara73,Yang98,Kitaoka99,CS99,GY00,Hanke04}. 
As an example, \cite{Siegel35}
gives an ingenious Gaussian sum technique to count solutions in case
$p$ does not divide $2t\det(\MQ)$.

The case of the prime $p=2$ is tricky and needs careful analysis. 
Pall \cite{Pall65} pointed out
that the work of Minkowski omits many details, resulting in errors 
for the case of prime~2. Later, Watson \cite{Watson76} found errors in 
the fixes suggested by Pall. 
It is believed by the community that the work by Watson
does not contain any errors.

\medskip

Coming back to our original problem of finding solutions, a few results
are known. We are aware of two relevant results on the question
of finding any solution (in contrast to sampling one, uniformly at random). The 
first \cite{AEM87,PS87} 
solves $x^2-ky^2\equiv m \pmod q$ for composite
$q$, when the factorization of $q$ is unknown. The second and more
relevant is the work done by Hartung \cite{Hartung08}. For odd $p$,
he gives a correct polynomial time algorithm to find one solution of 
$\Vx'\MQ\Vx\equiv t\bmod{p^k}$ (though it seems to be safe to say that
the possibility of this was folklore before). 
Unfortunately, his construction seems to contain
errors for the case $p=2$ (e.g., he divides by~2 in the ring $\ztkz$ while
proving Lemma~3.3.1 pp.~47-48). 

\paragraph*{Our Contribution.} Apart from the difficulty of giving
correct formulae for $p=2$, the method of
Minkowski (and others, including the Gaussian sum method) for 
counting the number of
solutions of $\Vx'\MQ\Vx\equiv t \bmod{p^k}$ has another drawback.
It is not constructive in the sense that it
does not provide a way to sample uniform solutions to the
equation. 
In this work, we give an alternate way of counting 
solutions, and thus by the above remarks, an
alternate way to compute the local density.
Our way of counting also yields a Las Vegas algorithm that, given an integral
quadratic form $\MQ$, a prime $p$, a positive integer $k$ and an integer
$t \in \zpkz$, runs in time $\poly(n,k,\log p)$ and samples a
uniform random solution of $\Vx'\MQ\Vx\equiv t\bmod{p^k}$, if a solution
exists.

\section{Preliminaries}

Integers and ring elements are denoted by lowercase letters,
vectors by bold lowercase letters and matrices by typewriter
uppercase letters. 
Sets are denoted by upper case letters $A, B, \cdots$, and their
cardinalities by $\fA, \fB, \cdots$. 
The $i$'th component of a vector $\Vv$ is
denoted by $v_i$. We use the notation $(v_1,\cdots,v_n)$ for
a column vector and the transpose of matrix $\MA$ is denoted by
$\MA'$. The matrix $\MA^n$ will denote a $n\times n$ square matrix.
The scalar product of two vectors will be denoted 
$\Vv'\Vw$ and equals $\sum_i v_iw_i$. The standard Euclidean norm of the 
vector $\Vv$ is denoted by $||\Vv||$ and equals $\sqrt{\Vv'\Vv}$.

If $\MQ_1^n, \MQ_2^m$ are matrices, then the {\em direct product}
of $\MQ_1$ and $\MQ_2$ is denoted by $\MQ_1\oplus\MQ_2$ and is
defined as $\diag(\MQ_1,\MQ_2)=\begin{pmatrix}
\MQ_1&0\\0&\MQ_2
\end{pmatrix}$.
Given two matrices $\MQ_1$ and $\MQ_2$ with the same number
of rows, $[\MQ_1,\MQ_2]$ is the matrix which is obtained
by concatenating the two matrices columnwise.
If $\MQ^n$ is a $n\times n$ integer matrix and $q$ is a positive integer
then $\MQ \bmod{q}$ is defined as the matrix with all entries of $\MQ$
reduced modulo $q$.
A matrix is called unimodular
if it is an integer $n\times n$ matrix with determinant $\pm 1$.

Let $\Ring$ be a commutative ring with unity and $\Ring^\times$ be the
set of units (i.e., invertible elements) of $\Ring$.
If $\MQ \in \Ring^{n\times n}$ is a square matrix, the 
{\em adjugate} of $\MQ$ is defined as the transpose of 
the cofactor matrix and is denoted by $\adj(\MQ)$. The matrix $\MQ$ 
is invertible if and only if $\det(\MQ)$ is a unit of $\Ring$. 
In this case, $\adj(\MQ)=\det(\MQ)\MQ^{-1}$. The set of invertible 
$n\times n$ matrices over $\Ring$ is denoted by $\gln(\Ring)$. The 
subset of matrices with determinant $1$ will be denoted by 
$\sln(\Ring)$.

\begin{fact}\label{fact:gln}
A matrix $\MU$ is in $\gln(\Ring)$ iff $\det(\MU) \in \Ring^\times$. 
\end{fact}

For every prime $p$ and positive integer $k$,
we define the ring $\zpkz=\{0,\cdots,p^k-1\}$, where product
and addition is defined modulo $p^k$.
Let $a, b$ be integers. Then, 
$\ordp(a)$ is the largest integer exponent of $p$ such that $p^{\ordp(a)}$
divides $a$. We let $\ordp(0) = \infty$. The $p$-coprime part of $a$ is then 
$\copp(a)=\frac{a}{p^{\ordp(a)}}$. Note that $\copp(a)$ is,
by definition, a unit of $\zpz$. 
For a positive integer $q$,
one writes $a\equiv b \bmod{q}$, if $q$ 
divides $a-b$. By $x:=a \bmod{q}$, we mean that
$x$ is assigned the unique value $b \in \{0,\cdots,q-1\}$ such that 
$b \equiv a \bmod{q}$.
An integer $t$ is 
called a {\em quadratic residue} modulo $q$ if $\gcd(t,q)=1$ and
$x^2\equiv t \bmod{q}$ has a solution. 

\begin{definition}\label{def:Legendre} 
Let $p$ be an odd prime, and $t$ be a positive integer with
$\gcd(t,p)=1$.
Then, the Legendre-symbol of $t$ with 
respect to $p$ is defined as follows.
\begin{displaymath}
\legendre{t}{p} = \left\{ \begin{array}{ll}
1 & \textrm{if $t$ is a quadratic residue 
 modulo $p$}\\
-1 & \textrm{otherwise.}
\end{array}\right.
\end{displaymath}
The Legendre symbol can also be computed by the Euler's 
formula, given by $t^{(p-1)/2} \bmod{p}$.
\end{definition}

For the prime~2, there is an extension of Legendre symbol called the
Kronecker symbol. It is defined for odd integers $t$ and 
$\legendre{t}{2}$ equals $1$ if $t\equiv \pm 1 \bmod 8$, and $-1$
if $t \equiv \pm 3 \bmod 8$.

The $p$-sign of $t$, denoted $\sgnp(t)$, is defined as 
$\legendre{\copp(t)}{p}$
for odd primes $p$ and $\copt(t) \bmod 8$ otherwise. We
also define $\sgnp(0)=0$, for all primes $p$.
Thus,
\[
\sgnp(0) = 0 \qquad 
\sgnp(t>0) \in \left\{
	\begin{array}{ll}
	\{+1,-1\} & \text{if $p$ is odd}\\
	\{1,3,5,7\} & \text{otherwise}
	\end{array}\right.
\]

The following lemma is well known.

\begin{lemma}\label{lem:QR}
Let $p$ be an odd prime. Then, there
are $\frac{p-1}{2}$ quadratic residues and $\frac{p-1}{2}$ 
quadratic non-residues modulo $p$. Also, every quadratic residue
in $\zpz$ can be written as a sum of two quadratic non-residues
and every quadratic non-residue can be written as a sum of two quadratic
residues.
\end{lemma}

An integer $t$ is a square modulo $q$ 
if there exists an integer $x$ such that $x^2\equiv t \pmod{q}$.
The integer $x$ is called the {\em square root} of $t$
modulo $q$. If no such $x$ exists, then $t$ is a non-square modulo $q$.

The following lemma is folklore and gives the necessary and sufficient
conditions for an integer $t$ to be a square modulo $p^k$. For
completeness, a proof is provided in Appendix \ref{sec:Proofs}.

\begin{lemma}\label{lem:Square}
Let $p$ be a prime, $k$ be a positive integer and 
$t \in \zpkz$ be a non-zero integer. Then, $t$ is a 
square modulo $p^k$ if and only if $\ordp(t)$ is even and
$\sgnp(t)=1$. 
\end{lemma}



\begin{definition}\label{def:Prim}
Let $p^k$ be a prime power. A vector $\Vv \in (\zpkz)^n$ is called 
primitive if there exists a component $v_i$, $i \in [n]$, of $\Vv$ such
that $\gcd(v_i,p)=1$. Otherwise, the vector $\Vv$ is 
non-primitive.
\end{definition}

Our definition of primitiveness of a vector is different but equivalent 
to the 
usual one in the literature. A
vector $\Vv \in (\zqz)^n$ is called primitive over $\zqz$ for a 
composite integer $q$ if it is primitive modulo $p^{\ordp(q)}$ for
all primes that divide $q$.

\begin{definition}\label{def:pexpansion}
Let $p$ be a prime, $k$ be a positive integer and $x$ be an element of 
$\zpkz$. The $p$-expansion of $x$ is $x$ written in base $p$ i.e.,
$x=x_{(0)} + x_{(1)} \cdot p + \cdots + x_{(k-1)}\cdot p^{k-1}$, where 
$x_{(i)} \in \zpz$ for $i \in \{0,\cdots,k-1\}$, is called the $i$'th
{\em digit} of $x$.
\end{definition}

\paragraph{Quadratic Form.} 
An $n$-ary quadratic form over a ring $\Ring$
is a symmetric matrix $\MQ \in \Ring^{n\times n}$, interpreted as the following
polynomial in $n$ formal variables $x_1,\cdots, x_n$ of uniform degree~2.
\[
\sum_{1\leq i,j \leq n}\MQ_{ij}x_ix_j = 
\MQ_{11}x_1^2 + \MQ_{12}x_1x_2 + \cdots = \Vx'\MQ\Vx
\]
The quadratic form is called {\em integral} if it is defined over the ring
$\bbZ$. It is called positive definite if for all non-zero column vectors
$\Vx$, $\Vx'\MQ\Vx > 0$. This work deals with integral quadratic forms,
henceforth called simply {\em quadratic forms}.
The {\em determinant} of the quadratic 
form is defined as $\det(\MQ)$. 
A quadratic form is called {\em diagonal} if $\MQ$ is a diagonal matrix. 

Given a set of formal variables 
$\Vx=\begin{pmatrix}x_1 & \cdots & x_n\end{pmatrix}'$ one can make a linear 
change of variables to $\Vy=\begin{pmatrix}y_1 & \cdots & y_n\end{pmatrix}'$ 
using a matrix $\MU \in \Ring^{n\times n}$ by setting $\Vy=\MU\Vx$. 
If additionally, 
$\MU$ is invertible over $\Ring$ i.e., $\MU \in \gln(\Ring)$, then this 
change of 
variables is reversible over the ring. We now define the equivalence of
quadratic forms over the ring $\Ring$ (compare with Lattice Isomorphism).

\begin{definition}\label{def:equiv}
Let $\MQ_1^n, \MQ_2^n$ be quadratic forms over a ring $\Ring$. They are called 
$\Ring$-{\em equivalent} if there exists a $\MU \in \gln(\Ring)$ such that 
$\MQ_2=\MU'\MQ_1\MU$.
\end{definition}

If $\Ring=\zqz$, for some positive integer $q$, then two integral
quadratic forms $\MQ_1^n$ and $\MQ_2^n$ will be called $q$-equivalent (denoted,
$\MQ_1\eqq \MQ_2$)
if there exists a matrix $\MU \in \gln(\zqz)$ such that 
$\MQ_2\equiv\MU'\MQ_1\MU\pmod q$.

If the equation $\Vx'\MQ\Vx \equiv t \pmod{q}$ has a solution
then we say that $t$ has a $q$-representation in $\MQ$ (or $t$ has 
a representation in
$\MQ$ over $\zqz$). Solutions 
$\Vx \in (\zqz)^n$ to the equation are called $q$-{\em representations} of
$t$ in $\MQ$. We classify the representations into two
categories: {\em primitive} and {\em non-primitive}, see Definition 
\ref{def:Prim}.
The set of non-primitive, primitive and all $p^k$-representations of $t$
in $\MQ$ is denoted by $\scpk(\MQ,t), \sbpk(\MQ,t)$ and
$\sapk(\MQ,t)$, respectively. Their sizes are denoted by 
$\fcpk(\MQ,t), \fbpk(\MQ,t)$ and $\fapk(\MQ,t)$, respectively. 

\paragraph{Randomized Algorithms.} 
Our randomized algorithms are 
Las Vegas algorithms. They either fail
and output nothing, or produce a correct answer. The 
probability of failure is bounded by a constant. Thus, for any
$\delta>0$, it is possible to repeat the algorithm 
$O(\log \frac{1}{\delta})$ times and succeed with probability at least
$1-\delta$. Henceforth, these algorithms will be called
{\em randomized algorithms}.

Our algorithms perform two kinds of operations. Ring operations e.g.,
multiplication, additions, inversions over $\zpkz$ and operations
over integers $\bbZ$ e.g., multiplications, additions, divisions etc.
The runtime for all these operations is treated as constant i.e., $O(1)$
and the time complexity of the algorithms is measured in terms of 
ring operations. 
For computing a uniform representation, we also need to sample a 
uniform ring element from $\zpkz$. We adapt the convention that sampling
a uniform ring elements also takes $O(1)$ ring operations.
For example, the Legendre symbol of an integer $a$ can be computed
by fast exponentiation in $O(\log p)$ ring operations over $\zpz$ 
while $\ordp(t)$
for $t \in \zpkz$ can be computed by fast exponentiation in $O(\log k)$ ring
operations over $\zpkz$.

\medskip

Let $\omega$ be the constant, such that multiplying two $n\times n$ 
matrices over $\zpkz$ takes
$O(n^{\omega})$ ring operations. 

\section{Technical Overview}

This section will give a complete outline of our method to count
solutions of $\Vx'\MQ\Vx \equiv t\bmod{p^k}$.

\subsection{Simplifying the LHS}

Recall the definition of equivalence of quadratic forms i.e.,
Definition \ref{def:equiv}. If $\MQ^n$ and $\MS^n$ are equivalent
over $\zpkz$ then the
following lemma show that the 
number of solutions of $\Vx'\MQ\Vx\equiv t \bmod{p^k}$ is the same as
the number of solutions
of $\Vx'\MS\Vx \equiv t \bmod{p^k}$.

\begin{lemma}\label{lem:QFEquivalence}
Let $p$ be a prime, $k, t$ be positive integers,
$\MQ^n$ be an integral quadratic form, $\MU \in \gln(\zpkz)$
and $\MS=\MU'\MQ\MU \modpk$. Then, 
$\fapk(\MQ,t)=\fapk(\MS,t)$, 
$\fbpk(\MQ,t)=\fbpk(\MS,t)$, and 
$\fcpk(\MQ,t) = \fcpk(\MS,t)$.
\end{lemma}
\begin{proof}
Let $\MV \in \gln(\zpkz)$ be such that $\MU\MV\equiv \MI \bmod{p^k}$.
The map $f:(\zpkz)^n \to (\zpkz)^n$ defined by 
$f(\Vx):= \MV\Vx \bmod{p^k}$ is bijective because $\MU, \MV$
are invertible over $\zpkz$.

If $\Vx$ is primitive then $\MV\Vx$ is also primitive. We prove this
by contradiction. Suppose $\MV\Vx$ is not primitive. Then, $\MV\Vx$
can be written as $p\Vy$, where $\Vy \in (\zpkz)^n$. But, $\Vx \equiv
\MU\MV\Vx \equiv p\MU \Vy \bmod{p^k}$, which implies that $\Vx$ is not
primitive.

The lemma now follows from the
equation $(\MV\Vx)'\MS(\MV\Vx) \equiv \Vx'\MQ\Vx \modpk$. 
\end{proof}

Thus, in order to count, we first transform $\MQ^n$ into a simpler
quadratic form. We refer to this transformation procedure as 
``diagonalization''.

An intuitive description of the diagonalization procedure 
follows. Given a quadratic form over a ring $\Ring$, one can classify
them according to the following equivalence. Two quadratic 
forms are equivalent over $\Ring$ if one can be obtained from the
other by an invertible linear change of variables over $\Ring$.
For example, $x^2$ and $2y^2$ are equivalent over the field of reals
$\bbR$ because
the transformations $x\to\sqrt{2}y$ and $y\to\frac{1}{\sqrt{2}}x$
are inverse of each other in $\bbR$, are linear and transform
$x^2$ to $2y^2$ and $2y^2$ to $x^2$ respectively. Thus, over $\bbR$
instead of trying to solve both $x^2$ and $2y^2$ separately, one 
can instead solve $x^2$ and then use the invertible linear 
transformation to map the solutions of $x^2$ to the solutions of $2y^2$.
It is well known that every quadratic form in $n$-variables over 
$\bbR$ is equivalent to $\sum_{i=1}^a x_i^2 - \sum_{i=a+1}^n x_i^2$, 
for some $a \in [n]$. This is known as the Sylvester's Law of inertia.

For the ring $\zpkz$ such that $p$ is odd, there
always exists an equivalent quadratic form which is also diagonal 
(see \cite{CS99}, Theorem 2, page 369).
Additionally, one can explicitly find the invertible change of
variables that turns it into a diagonal quadratic form.
The situation is tricky over the ring $\ztkz$. Here, it might not
be possible to eliminate all mixed terms, i.e., terms of the form
$2a_{ij}x_ix_j$ with $i\neq j$. For example, consider
the quadratic form $2xy$ i.e., 
$\begin{pmatrix}0 & 1 \\ 1 & 0 \end{pmatrix}$ over $\ztkz$. 
An invertible linear change of variables over $\ztkz$ is
of the following form.
\begin{gather*}
\begin{array}{l}x \to a_1x_1 + a_2x_2\\ 
y \to b_1x_1+b_2x_2\end{array} \qquad 
\begin{pmatrix}a_1 & a_2 \\ b_1 & b_2\end{pmatrix} \text{ invertible over
$\ztkz$}
\end{gather*}
The mixed term after this transformation is $2(a_1b_2+a_2b_1)$. 
As $a_1b_2+a_2b_1 \bmod 2$ is the same as the determinant
of the change of variables above i.e., $a_1b_2-a_2b_1$
modulo $2$; it is not possible for a transformation in 
$\gl_2(\ztkz)$ to eliminate the mixed term. 
Instead, one can show that over $\ztkz$ it is possible to get an 
equivalent form where the mixed terms are disjoint i.e., both
$x_ix_j$ and $x_ix_k$ do not appear, where $i,j,k$ are 
pairwise distinct.
One captures this form by the following definition.

\begin{definition}\label{def:BlockDiagonal}
A matrix $\MD^n$ over integers is in a block diagonal form if it is a
direct sum of type I and type II forms; where type I form is an integer
while type II is a matrix of the form 
$\begin{pmatrix}2^{\ell+1}a & 2^{\ell} b \\ 2^{\ell}b & 2^{\ell+1}c\end{pmatrix}$
with $b$ odd.
\end{definition}

The following theorem is folklore and is also implicit in the proof of
Theorem~2 on page~369 in \cite{CS99}. For completeness, we provide a 
proof in Appendix \ref{sec:BlockDiagonal}.

\begin{theorem}\label{thm:BlockDiagonal}
Let $\MQ^n$ be an integral quadratic form, $p$ be a prime, 
and $k$ be a positive integer. 
Then, there is an algorithm that performs $O(n^{1+\omega}\log k)$ ring operations
and produces a matrix $\MU \in \sln(\zpkz)$ such that $\MU'\MQ\MU\pmod{p^k}$, is a
diagonal matrix for odd primes $p$ and a block diagonal matrix
(in the sense of Definition~\ref{def:BlockDiagonal}) for $p=2$.
\end{theorem}

\subsection{Simplifying the RHS}

Our next step is to simplify the right hand side of the equation
$\Vx'\MQ\Vx \equiv t \bmod{p^k}$ i.e., $t$. 

\begin{definition}\label{def:SymbolOneDim}
The $p^k$-{\em symbol} of an integer $t$ is
$\sympk(t)=(\ordp(t \modpk), \sgnp(t \modpk))$.
\end{definition}

A $p^k$-symbol will be denoted as $\gamma$ and $\ordp(\gamma)$
will denote the $p$-order of $\gamma$ and $\sgnp(\gamma)$ will denote the
$p$-sign of $\gamma$.
The next lemma shows the importance of the $p^k$-symbol.

\begin{lemma}\label{lem:IntegerSymbolIsInvariant}
For integers $a,b$ and prime $p$:
$b \eqpk a$ iff $\sympk(a)=\sympk(b)$.
\end{lemma}
\begin{proof}
The lemma is true if $\ordp(a)$ or $\ordp(b)$ is at least
$k$. Hence, we assume that $\ordp(a), \ordp(b) < k$.

We first show that $b \eqpk a$ implies $\sympk(a)=\sympk(b)$.
If $b \eqpk a$ then there 
exists a $u \in (\zpkz)^\times$ such that 
$b \equiv u^2a \pmod{p^k}$. But, multiplying by a square of a unit does
not change the sign i.e., 
$\sgnp(a \bmod{p^k})=\sgnp(u^2a \bmod p^k)=\sgnp(b \bmod{p^k})$.
Also, $\ordp(u)=0$ implies that $\ordp(a)=\ordp(b)$. This shows
that $\sympk(a)=\sympk(b)$.

We now show the converse. Suppose $a$ and $b$ be such that 
$\sympk(a)=\sympk(b)$. Let $\ordp(a)=\ordp(b)=\alpha$. By
definition of $p^k$-symbol, $\sgnp(a \modpk)=\sgnp(b \modpk)$. But then,
\begin{align*}
\sgnp\left(\copp(a) \bmod{p^{k-\alpha}}\right) = 
\sgnp\left(\copp(b) \bmod{p^{k-\alpha}}\right)\\
\iff \sgnp\left(\copp(a)\copp(b)^{-1} \bmod{p^{k-\alpha}}\right) = 1
\end{align*}
By Lemma \ref{lem:Square}, 
$\copp(a)\copp(b)^{-1} \bmod{p^{k-\alpha}}$ is a 
quadratic residue modulo $p^{k-\alpha}$. But then,
there exists a unit $u$ such that 
\[
u^2 \equiv \copp(a)\copp(b)^{-1} \pmod{p^{k-\alpha}}\;.
\] 
Multiplying this equation by 
$\copp(b)p^\alpha$ yields
$u^2b \equiv a \modpk$ or $b \eqpk a$.
\end{proof}

The following lemma shows that the number of 
solutions only depend on two things: $\ordp(t \modpk)$, and 
$\sgnp(t \modpk)$.

\begin{lemma}\label{lem:SymbolInvariance}
Let $\MQ^n$ be a quadratic form, $p$ be a prime, $k$ be a positive
integer and $t,s$ be integers such that $\sympk(t)=\sympk(s)$.
Then, $\fapk(\MQ,t)=\fapk(\MQ,s)$, 
$\fbpk(\MQ,t)=\fbpk(\MQ,s)$, and 
$\fcpk(\MQ,t)=\fcpk(\MQ,s)$.
\end{lemma}

\begin{proof}
By Lemma \ref{lem:IntegerSymbolIsInvariant}, it follows 
that there exists a unit $u \in (\zpkz)^\times$
such that $s\equiv u^2t \modpk$. But then, 
$t$ has a representation $\Vx \in (\zpkz)^n$ in $\MQ$ iff $s$ has 
a representation $u\Vx \in (\zpkz)^n$ in $\MQ$; 
\[
(u\Vx)'\MQ(u\Vx)\equiv u^2 \Vx'\MQ\Vx \equiv u^2t \equiv s \modpk\;.
\]
The function $\Vx\to u\Vx$ maps $p^k$-representations of $t$ in $\MQ$
to $p^k$-representations of $s$ in $\MQ$.
Also, the map is bijective, and preserves primitiveness;
completing the proof.
\end{proof}

Let $\sym$ be a $p^k$-symbol and $t \in \zpkz$ be an
integer such that $\sympk(t)=\sym$. Then,
using Lemma \ref{lem:SymbolInvariance}, we can define the following
quantities. 
\[
\fapk(\MQ,\sym) = \fapk(\MQ,t) \qquad
\fbpk(\MQ,\sym) = \fbpk(\MQ,t) \qquad
\fcpk(\MQ,\sym) = \fcpk(\MQ,t)
\]

There are $p^k$ 
different possible values for $t$ over $\zpkz$ i.e., exponential
in $k$. But, for the $p^k$-symbol $\sym$, there are only 
$(4k+1)$ possibilities; when $p=2$, and $(2k+1)$ otherwise (the 
``$+1$'' is for $0$). Note that the $p$-order $\infty$ only appears
with $p$-sign~0 and vice-versa.

For notational convenience, we define the following sets (the modulo $p^k$
will be clear from the context, whenever we use this notation).
\begin{align}\label{def:opk}
\ORD = \{\infty,0,\cdots,k-1\} ~~
\SGN = \left\{\begin{array}{ll}
\{0,1,-1\} & \text{$p$ is an odd prime} \\
\{0,1,3,5,7\} & \text{otherwise}
\end{array}\right.
\end{align}

The following definition is useful in reducing the problem of 
counting representations in higher dimensions to the problem of
counting representations for individual blocks in a block diagonal
form.

\begin{definition}\label{def:SymbolSplitSize}
Let $p$ be a prime, $k$ be a positive integer, 
$t \in \zpkz$ be an integer, and $\gamma_1, \gamma_2$ be 
$p^k$-symbols.
Then, the $(\gamma_1,\gamma_2)$-split size of $t$ over $\zpkz$,
denoted $\fspkt(\gamma_1,\gamma_2)$, is the size of the following
set,
\[
\begin{aligned}
\spkt(\gamma_1,\gamma_2) = \Big\{(a,b) \in (\zpkz)^2\mid 
\sympk(a)=\gamma_1, \sympk(b)=\gamma_2, 
\hfill t\equiv a+b \modpk\Big\}.
\end{aligned}
\]
\end{definition}

Let $\sym, \sym_1, \sym_2$ be $p^k$-symbols and $t \in \zpkz$ be an
integer such that $\sympk(t)=\sym$. Then, the following
lemma shows that we can define 
$\fS^{\sym}_{p^k}(\sym_1,\sym_2)$ as
$\fS^{t}_{p^k}(\sym_1,\sym_2)$.

\begin{lemma}\label{lem:SplitClassInvariance}
Let $p$ be a primes, $k$ be a positive integer, 
$\sym_1, \sym_2$ be $p^k$-symbols and $t,s$ be 
integers such that $\sympk(t)=\sympk(s)$. Then, 
$\fS^t_{p^k}(\sym_1,\sym_2)=
\fS^s_{p^k}(\sym_1,\sym_2)$.
\end{lemma}
\begin{proof}
By Lemma \ref{lem:IntegerSymbolIsInvariant}, there exists
a $u \in (\zpkz)^\times$ such that $u^2t\equiv s \modpk$.
If $t \equiv a + b \modpk$ then 
$u^2t \equiv u^2a+u^2b \modpk$ with $\sympk(u^2a)=\sympk(a)$,
and $\sympk(u^2b)=\sympk(b)$. The lemma now follows from the
observation that the map $x \to u^2 x$ is a bijection from $\zpkz$ to
itself.
\end{proof}

\begin{lemma}\label{lem:Split}
Let $\MQ=\diag(\MQ_1,\MQ_2)$ be an integral quadratic form, 
$p$ be a prime, $k$ be a positive integer,
and $\sym$ be a $p^k$-symbol. Then,
\begin{gather*}
\fapk(\MQ,\sym) = \sum_{\gamma_1, \gamma_2\in\ORD\times\SGN} 
	\fS^{\sym}_{p^k}(\gamma_1,\gamma_2)\cdot\fapk(\MQ_1,\gamma_1)
	\cdot\fapk(\MQ_2,\gamma_2)\\
\fcpk(\MQ,\sym) = \sum_{\gamma_1, \gamma_2 \in \ORD\times \SGN} 
	\fS^{\sym}_{p^k}(\gamma_1,\gamma_2)\cdot\fcpk(\MQ_1,\gamma_1)
	\cdot\fcpk(\MQ_2,\gamma_2) 
\end{gather*}
\end{lemma}
\begin{proof}
Let $t \in \zpkz$
be such that $\sympk(t)=\sym$.
The formula for the total number of representations of $\sym$ by 
$\MQ$ over $\zpkz$ follows from the calculations below. 
\begin{align*}
\fapk(\MQ,\sym) &=\fapk(\MQ,t) 
= \sum_{a \in \zpkz} \fapk(\MQ_1,a)\cdot\fapk(\MQ_2,t-a) \\
&= \sum_{a \in \zpkz} \fapk(\MQ_1,\sympk(a))\cdot\fapk(\MQ_2,
	\sympk(t-a)) \\
&= \sum_{\gamma_1,\gamma_2 \in \ORD\times\SGN} 
\fS^t_{p^k}(\gamma_1,\gamma_2)\cdot
	\fapk(\MQ_1,\gamma_1) \cdot\fapk(\MQ_2,\gamma_2) \\
&= \sum_{\gamma_1,\gamma_2 \in \ORD\times\SGN} 
\fS^{\sym}_{p^k}(\gamma_1,\gamma_2)\cdot
	\fapk(\MQ_1,\gamma_1) \cdot\fapk(\MQ_2,\gamma_2) 
\end{align*}

The same
calculation works for the number of non-primitive representations
because a representation of $\sym$ by $\MQ$ is non-primitive iff
every component of the representation is non-primitive. 
\end{proof}

\subsection{Overview of the Counting/Sampling Algorithm}

Given $(\MQ^n,p,k,t)$ our 
counting algorithm for finding $\fapk(\MQ,t)$ is analyzed in detail in
Section \ref{sec:Count}. An overview of the algorithm is given below.
\begin{enumerate}
\addtolength{\itemsep}{-4pt}
\item Block diagonalize $\MQ$ over $\zpkz$ using Theorem 
\ref{thm:BlockDiagonal}. Let $\MD^n=\MD_1 \oplus \cdots \oplus \MD_m$
be the block diagonal form returned by the algorithm. Recall,
each $\MD_i$ is either Type I i.e., an integer, or Type II (only
when $p=2$).
\item For each symbol $\sym \in \ORD \times \SGN$ and $i \in [m]$,
calculate $\fapk(\MD_i, \sym)$. The case of prime~2 is handled 
separately and needs careful analysis for Type II blocks. (Sections
\ref{sec:Dim1P}, \ref{sec:Dim12}, and \ref{sec:Dim2II})
\item For each triple $\sym, \sym_1, \sym_2 \in \ORD\times\SGN$ compute the 
size of split classes (Section \ref{sec:SplitSize})
i.e., $\mathcal{S}^{\sym}_{p^k}(\sym_1,\sym_2)$. 
\item Compute 
$\fapk(\MD_1 \oplus \cdots \oplus \MD_i, \sym)$ for each 
$\sym \in \ORD\times\SGN$ and $i \in [m]$, using Lemma 
\ref{lem:Split}.
\item Output $\fapk(\MD,\sympk(t))$. 
\addtolength{\itemsep}{4pt}
\end{enumerate}

As mentioned in the introduction of the paper, this algorithm
can also be used to compute what mathematicians call the ``local density''
(see Section \ref{Sec:Density}).
Furthermore, this algorithm can be generalized to sample uniform representations
(details in Section \ref{Sec:Sample}). 
The following two theorems are the main contribution of this paper
(proved in Section \ref{Sec:SamplingTheorem}). 

\begin{theorem}\label{thm:Sample2}
Let $\MQ^n$ be an integral quadratic form, $k$ be a positive integer, 
and $t$ be an element of $\ztkz$. Then, there exists a deterministic 
polynomial time algorithm that performs
$O(n^{1+\omega}\log k+nk^3)$ ring operations over $\ztkz$ and
samples a uniform (primitive/non-primitive)
representation of $t$ by $\MQ$ over $\ztkz$, 
if such a representation exists.
\end{theorem}

\begin{theorem}\label{thm:SampleP}
Let $\MQ^n$ be an integral quadratic form, $p$ be an odd prime, 
$k$ be a positive integer, $t$ be an element of $\zpkz$.
Then, there is a 
polynomial time algorithm Las Vegas algorithm
that performs
$O(n^{1+\omega}\log k+nk^3+n\log p)$ ring operations over 
$\zpkz$ and fails with constant probability (say, at most $\frac1{3}$).
Otherwise, the algorithm
outputs a uniform  (primitive/non-primitive)
 $p^k$-representation of $t$ by $\MQ$,
if such a representation exists.
\end{theorem}

In other words, the algorithm is able to output a uniform representation,
a representation which is uniform among the primitive ones, and a 
representation which is uniform among the non-primitive ones.

\section{Squares over Quotient Rings}\label{PS:sec:squares}

To understand quadratic forms and their equivalence, we need some
elementary results from number theory and algebra. In particular,
we want to know when an integer $t$ is a square over the ring $\zpkz$.

The prime~2 creates some technical complications. Thus, we chose to
spilt the results in two parts; for odd primes, and for the even prime.
We remark that everything in this section is well-known. 


\subsection{Squares Modulo $p^k$; $p$ odd}
A positive
integer $t$ is a quadratic residue modulo $p^k$ iff $t$ is a 
quadratic residue modulo $p$ (see Theorem $2.30$, \cite{Shoup09}).
If $t$ is a quadratic residue modulo $p$ then there are two 
solutions to the equation $x^2\equiv t \bmod{p}$. Furthermore, 
there is a Las-Vegas algorithm due to Cippola 
and Lehmer which performs $O(\log p)$ ring operations over $\zpz$ and 
computes the square root of a quadratic residue $t$ over $\zpz$.

Given a positive integer $t$ such that $t$ is a quadratic residue 
modulo $p$ both solutions to $x^2\equiv t \bmod{p^k}$ can be 
found using Hensel's Lemma. We give a sketch of the proof here and
a full proof can be found on page~174, \cite{BS96}.

\begin{lemma}\label{lem:HenselP}
Let $p$ be an odd prime, $k$ a positive integer, and $t$ be a quadratic residue
modulo $p$. Then, there is a randomized algorithm that performs 
$O(\log k+\log p)$ ring operations and with constant probability outputs 
both solutions of $x^2\equiv t \bmod{p^k}$.
\end{lemma}
\begin{proof}
We use the Las Vegas algorithm by Cippola and Lehmer to first find
the solutions of $x^2\equiv t \bmod{p}$. We now show how to ``lift''
this solution from $\zpz$ to $\zpkz$.
We do this incrementally by lifting solutions of 
$x^2 \equiv t \bmod{p^e}$ to
solutions of $x^2\equiv t \bmod{p^{2e}}$. 

Let $a$ be a solution of $x^2\equiv t \bmod{p^e}$ i.e., 
$a^2 \equiv t \bmod{p^e}$. If
$b := \frac{(t-a^2)/p^e}{2a} \modpk$, then
\[
(a+p^eb)^2 \equiv a^2+2abp^e \equiv t \pmod{p^{2e}} \;.
\]
The computation of $b$ takes $O(1)$ ring operations and hence to
lift a solution modulo $p$ to a solution modulo $p^k$, we need to
perform $O(\log k)$ ring operations.
\end{proof}

\subsection{Squares Modulo $2^k$}

There is only one non-zero element in $\ztz$ i.e., $1$ and by definition,
it is a quadratic residue. This is the only quadratic residue 
in $\bbZ/4\bbZ$ and $\bbZ/8\bbZ$. 



%

Let $t \in \ztkz$ be an integer. Then, $t$ is a quadratic residue modulo
$2^k$ iff $t\equiv 1 \bmod 8$ (see Theorem~2, pp~49, \cite{BS86}).
A similar result as in Lemma \ref{lem:HenselP} can be shown. The ``lifting''
of solutions modulo $p^e$ to solutions modulo $p^{2e}$ does not work
in this case because~2 does not have an inverse (see the proof of Lemma
\ref{lem:HenselP}). 

\begin{lemma}\label{lem:Hensel2}
Let $k$ be a positive integer, and $t$ be an element of $\ztkz$ such that
$t\equiv 1\pmod{8}$. Then, $x^2\equiv t \pmod{2^k}$ has one solution
for $k=1$, two for $k=2$ and four otherwise. Additionally, all solutions
can be found in $O(k)$ ring operations.
\end{lemma}
\begin{proof}
The set of solutions is $\{1\}$ in $\ztz$, $\{1,3\}$ in $\bbZ/4\bbZ$ 
and $\{1,3,5,7\}$ in $\bbZ/8\bbZ$.
For larger $k$, we prove the lemma by induction. 

Let us assume that the lemma is true for $k\geq 3$ and let $b$ be a square root
of $t$ modulo $2^k$. Then, $b^2 \equiv t \modtk$ and $2^k|(t-b^2)$. By 
Definition \ref{def:pexpansion}, the $k$'th digit of $t-b^2$ is 
$d=(t-b^2)/2^k \bmod 2$. Recall, $x_{(k)}$ is the $k$'th digit in the 
$p$-expansion of $x$ (see Definition \ref{def:pexpansion}, $p=2$ here). Let
$c=b+2^{k-1}d$. From $b\equiv 1 \pmod 2$, we conclude that 
$2^kbd \equiv 2^kd \pmod{2^{k+1}}$. For $k\geq 3$, it follows that,
\begin{align*}
c^2&\equiv b^2+2^kbd \equiv 
(\sum_{i=0}^{k-1}t_{(i)}2^i + (b^2)_{(k)}2^{k})
+2^kd \equiv t\pmod{2^{k+1}}
\end{align*}
Additionally, $-c,2^k+c,2^k-c$ are all distinct
solutions of $x^2\equiv t \pmod{2^{k+1}}$. 
The fact that these are the only possible solutions is argued as follows.
If $x$ is a solution modulo $2^{k+1}$ 
then it is also a solution modulo $2^k$. So, there are only eight 
choices for $x$ in $\bbZ/2^{k+1}\bbZ$ given four solutions modulo $\ztkz$.
Only four of these work.
The proof is constructive i.e.,
given all solutions modulo $2^{k}$, all solutions modulo 
$2^{k+1}$ can be found in constant number of ring operations.
\end{proof}

\section{Counting Representations}

In this section, we count the number of solutions of the equation
$\Vx'\MQ\Vx \equiv t \modpk$. 

\subsection{Dimension = 1, Odd Prime}\label{sec:Dim1P}

The following lemma gives the necessary and sufficient conditions for
an integral quadratic form $\MQ^{n=1}$ to represent $t$ over $\zpkz$,
when $k > \ordp(t)$.

\begin{lemma}\label{lem:RepresentTModPk1D}
If $\MQ, t$ be integers, $p$ be an odd prime and $k>\ordp(t)$.
Then, $t$ can be represented by $\MQ$ over $\zpkz$ if and only 
if $\ordp(t)-\ordp(\MQ)$ is even, $\geq 0$, and 
$\legendre{\copp(\MQ)\copp(t)}{p}=1$. 
\end{lemma}
\begin{proof}
Suppose that $t$ can be represented by $\MQ$ over $\zpkz$ and $k>\ordp(t)$.
Then, there exists integers $x$ and $a$ such that $x^2\MQ=t+ap^k$. But then,
\begin{align*}
\ordp(t) = \ordp(t+ap^k) =  \ordp(x^2\MQ) = 2\ordp(x)+\ordp(\MQ) \;,\\
\legendre{\copp(t)}{p} = \legendre{\copp(t+ap^k)}{p} 
= \legendre{\copp(x^2\MQ)}{p} = \legendre{\copp(\MQ)}{p} \;.
\end{align*}

Conversely, suppose that $\ordp(t)-\ordp(\MQ)$ is even, $\geq 0$ and 
$\legendre{\copp(\MQ)\copp(t)}{p}=1$. Then, by Lemma 
\ref{lem:IntegerSymbolIsInvariant}, there exists an integer $u$ such 
that $\copp(\MQ)u^2 \equiv \copp(t) \pmod{p^k}$. But then, multiplying
this equation by $p^{\ordp(t)}$, we conclude that 
$x=p^{\frac{\ordp(t)-\ordp(\MQ)}{2}}u$ is a $p^k$-representation of $t$; 
as follows.
\[
(p^{\frac{\ordp(t)-\ordp(\MQ)}{2}}u)^2\MQ = p^{\ordp(t)}\copp(\MQ)u^2 
\equiv p^{\ordp(t)}\copp(t) \pmod{p^k}\;.
\]
\end{proof}

\begin{lemma}\label{lem:CountModPk1D}
Let $\MQ, t, k$ be integers, and $p$ be an odd prime. Then, 
Algorithm \ref{alg:CountModPk1D} performs $O(\log k+\log p)$ ring operations 
and outputs the 
number of primitive and non-primitive $p^k$-representations of $t$ in $\MQ$.
\end{lemma}
\begin{proof}
We want to count the number of primitive and non-primitive 
$x \in \zpkz$ such that $x^2\MQ\equiv t \pmod{p^k}$. We 
distinguish between the following cases.

\paragraph{$\mathbf{\ordp(t) \geq k}$.} In this case, $t \bmod{p^k}$
is $0$. Thus,
\[
x^2p^{\ordp(\MQ)}\equiv 0 \pmod{p^k}
\iff p^{k-\ordp(\MQ)}|x^2 
\iff p^{\lceil \frac{k-\ordp(\MQ)}{2}\rceil}|x
\]
There are primitive representations iff $\ordp(\MQ) \geq k$. But then,
every $x \in \zpkz$ is a representation. 
Recall, the $p$-expansion of $x \in \zpkz$ (Definition \ref{def:pexpansion}).
By definition, $x$ is primitive iff $x_{(0)} \neq 0$. The rest of the 
$k-1$ digits can be chosen freely. The number of
primitive and non-primitive representations in the case of 
$\ordp(\MQ) \geq k$ is then $(p-1)p^{k-1}$ and $p^{k-1}$ respectively.
Otherwise, if $\ordp(\MQ) < k$ then there are no primitive representations
and the number of non-primitive representations is 
$p^{k-\lceil \frac{k-\ordp(\MQ)}{2}\rceil}$. This completes the case
of $\ordp(t) \geq k$.

\paragraph{$\mathbf{\ordp(t) < k}$.} It follows from
Lemma \ref{lem:RepresentTModPk1D} that $\MQ$ represents $t$ over $\zpkz$
iff $\ordp(t)-\ordp(\MQ) \geq 0$, is even and 
$\legendre{\copp(t)\copp(\MQ)}{p}=1$. But then,
\begin{align*}
x^2\MQ \equiv &t \pmod{p^k} \iff 
x^2 \equiv p^{\ordp(t)-\ordp(\MQ)}\copp(t)\copp(\MQ)^{-1} 
\pmod{p^{k-\ordp(\MQ)}} \\
&\iff x=p^{\frac{\ordp(t)-\ordp(\MQ)}{2}}y, 
\qquad y \in \bbZ/p^{k-\frac{\ordp(t)-\ordp(\MQ)}{2}}\bbZ, \\
&\qquad y^2 \equiv \copp(t)\copp(\MQ)^{-1}  \pmod{p^{k-\ordp(t)}}\;.
\end{align*}
The number of possible
representations is the number of $y \in \bbZ/p^{k-\frac{\ordp(t)-\ordp(\MQ)}{2}}\bbZ$ 
satisfying
$y^2 \equiv \copp(t)\copp(\MQ)^{-1} \bmod{p^{k-\ordp(t)}}$. As 
$\copp(t)\copp(\MQ)^{-1}$ is a quadratic residue modulo $p$, by Lemma
\ref{lem:HenselP}, there are exactly two possible $y$ over the ring
$\bbZ/p^{k-\ordp(t)}\bbZ$. 
Recall, the $p$-expansion of $y$
(definition \ref{def:pexpansion}). In the $p$-expansion of $y$, 
there are $k-\frac{\ordp(t)-\ordp(\MQ)}{2}$ digits; the first
$k-\ordp(t)$ of those must be a solution of 
$y^2=\copp(t)\copp(\MQ)^{-1}$ over the ring 
$\bbZ/p^{k-\ordp(t)}\bbZ$. Hence, the remaining 
$\frac{\ordp(t)+\ordp(\MQ)}{2}$ can be chosen freely from $\zpz$.
Thus, the number of $p^k$-representation of $t$ in $\MQ$ is
$2p^{\frac{\ordp(t)+\ordp(\MQ)}{2}}$. Note that
there are primitive representations iff $\ordp(t)=\ordp(\MQ)$. But then,
every representation $x \in \zpkz$ is primitive. 
This completes the case of
$\ordp(t)<k$.

\begin{algorithm}
\SetKw{or}{or}
\SetKwData{prim}{prim}
\SetKwData{nprim}{nprim}
\SetKw{Return}{return}
\SetKwInOut{Input}{input}
\caption{\textsc{CountModpk1Dim}($\MQ^{n=1}, \sympk(t), p, k$)}\label{alg:CountModPk1D}
\uIf{$\ordp(t) \geq k$}{\nllabel{if:count1dmodp}
	\lIf{$\ordp(\MQ)\geq k$}{
		\prim$:=(p-1)p^{k-1}$;
		\nprim$:=p^{k-1}$
	}
	\lElse{
		\prim$:=0$;
		\nprim$:=p^{k-\lceil (k-\ordp(\MQ))/2 \rceil}$
	}
}
\lElseIf{$\ordp(t)-\ordp(\MQ) < 0$ \or $\ordp(t)-\ordp(\MQ)$ odd 
		\or $\legendre{\copp(t)\copp(\MQ)}{p}=-1$}{\nllabel{else:count1dmodp}
	\prim$:=0;$ \nprim$:=0$
}
\lElseIf{$\ordp(\MQ)=\ordp(t)$}{
	\prim$:=2p^{\ordp(\MQ)};$ 
	\nprim$:=0$
}
\lElse{
	\prim$:=0;$ 
	\nprim$:=2p^{(\ordp(t)+\ordp(\MQ))/2}$
}
\end{algorithm}

The algorithm works with numbers of size $p^k$ and makes only a
constant number of operations. The Legendre symbol can be 
calculated using fast exponentiation (see Definition \ref{def:Legendre}) 
in $O(\log p)$ ring operations and the computations of $p$-orders take 
$O(\log k)$ ring operations. 
The algorithm hence, performs at most $O(\log k+\log p)$ ring operations.
\end{proof}

\subsection{Dimension $=1$, $p=2$}\label{sec:Dim12}

The following lemma gives the necessary and sufficient conditions for 
an integral quadratic form $\MQ^{n=1}$ to represent a non-zero 
$t \in \ztkz$ over $\ztkz$.

\begin{lemma}\label{lem:RepresentTMod2k1D}
Let $\MQ$ be an integer, and $t$ be a non-zero integer from 
$\ztkz$.
Then, $t$ can be represented by $\MQ$ over $\ztkz$ if and only 
if $\ordt(t)-\ordt(\MQ)$ is even, $\geq 0$, and 
$\copt(\MQ)\equiv \copt(t)\bmod{\min\{8,2^{k-\ordt(t)}\}}$.
\end{lemma}
\begin{proof}
Suppose that $t$ can be represented by $\MQ$ over $\ztkz$ and 
$k>\ordt(t)$. Then, there exists integers $x$ and $a$ such that
$x^2\MQ=t+a2^k$. But then,
\begin{gather*}
\ordt(t)=\ordt(t+a2^k)=\ordt(x^2\MQ)=2\ordt(x)+\ordt(\MQ)\;,\\
\copt(t+a2^k)=\copt(x^2\MQ)=\copt(x)^2\copt(\MQ)\equiv \copt(\MQ) \pmod{8}\;.
\end{gather*}
By assumption, $\ordt(t)<k$ and hence 
$\copt(t+a2^k)\equiv\copt(t) \bmod{2^{k-\ordt(t)}}$.

Conversely, suppose that $\ordt(t)-\ordt(\MQ)$ is even, $\geq 0$ and
$\copt(t)\equiv \copt(\MQ) \bmod \min\{8,2^{k-\ordt(t)}\}$. Then, by Lemma 
\ref{lem:IntegerSymbolIsInvariant}, there exists an integer $u$ such
that $\copt(\MQ)u^2\equiv \copt(t) \bmod{2^{k-\ordt(t)}}$. But then, multiplying
this equation by $2^{\ordt(t)}$, we conclude that 
$x=2^{\frac{\ordt(t)-\ordt(\MQ)}{2}}$ is a $2^k$-representation of $t$; 
as follows.
\[
(2^{\frac{\ordt(t)-\ordt(\MQ)}{2}}u)^2\MQ = 2^{\ordt(t)}\copt(\MQ)u^2 
\equiv 2^{\ordt(t)}\copt(t) \modtk\;.
\]
\end{proof}


\begin{lemma}\label{lem:CountMod2k1D}
Let $\MQ, t,$ and $k$ be integers. Algorithm 
\ref{alg:CountMod2k1D} performs $O(\log k)$ ring operations and outputs 
the number of primitive and non-primitive $2^k$-representations
of $t$ in $\MQ$.
\end{lemma}
\begin{proof}
We want to count the number of primitive and non-primitive 
$x \in \ztkz$ such that $x^2\MQ\equiv t \modtk$. We 
distinguish between the following cases.

\begin{algorithm}
\SetKw{or}{or}
\SetKwData{nprim}{non-prim}
\SetKwData{prim}{prim}
\SetKwData{rep}{rep}
\SetKw{Return}{return}
\SetKwInOut{Input}{input}
\caption{\textsc{CountMod2k1Dim}($\MQ^{n=1}, k, \symtk(t)$)}
\label{alg:CountMod2k1D}
\uIf{$\ordt(t)\geq k$}{
	\lIf{$\ordt(\MQ)\geq k$}{
		\prim$:=2^{k-1}$;
		\nprim$:=2^{k-1}$
	}
	\lElse{
		\prim$:=0$;
		\nprim$:=2^{k-\lceil (k-\ordt(\MQ))/2 \rceil}$
	}
}
\lElseIf{$\ordt(t)-\ordt(\MQ) < 0$ \or $\ordt(r)-\ordt(\MQ)$ odd 
		\or $\cop_2(r)\not\equiv \cop_2(\MQ) \bmod{\min\{8,2^{k-\ordt(t)}\}}$}{
	\prim$:=0;$ \nprim$:=0$
}
\Else{
	\lIf{$k-\ordt(t)\geq 3$}{$\rep:=4\cdot 2^{(\ordt(t)+\ordt(\MQ))/2}$}
	\lElse{$\rep:=(k-\ordt(t))\cdot 2^{(\ordt(t)+\ordt(\MQ))/2}$}
	\lIf{$\ordt(\MQ)=\ordt(t)$}{$\prim:=\rep; \nprim:=0$}
	\lElse{$\prim:=0; \nprim:=\rep$}
}
\Return{$[\prim, \nprim]$}
\end{algorithm}

\paragraph{$\mathbf{\ordt(t) \geq k}$.} In this case, $t \modtk$
is $0$ and hence,
\[
x^22^{\ordt(\MQ)} \equiv 0 \modtk
\iff 2^{k-\ordt(\MQ)}|x^2 
\iff 2^{\lceil \frac{k-\ordt(\MQ)}{2}\rceil}|x
\]
There are primitive representations iff $\ordt(\MQ) \geq k$. But then,
every $x \in \ztkz$ is a representation. 
Recall, the $2$-expansion of $x \in \ztkz$ (definition \ref{def:pexpansion}).
By definition, $x$ is primitive iff $d_0(x)=1$. The rest of the 
$k-1$ digits can be chosen freely. The number of
primitive and non-primitive representations in the case of 
$\ordt(\MQ) \geq k$ is then $2^{k-1}$ and $2^{k-1}$ respectively.
Otherwise, if $\ordt(\MQ) < k$ then there are no primitive representations
and the number of non-primitive representations is 
$2^{k-\lceil \frac{k-\ordt(\MQ)}{2}\rceil}$. This completes the case
of $\ordt(t) \geq k$.

\paragraph{$\mathbf{\ordt(t) < k}$.} In this case, 
$t$ is a non-zero element of $\ztkz$.
It follows from
Lemma \ref{lem:RepresentTMod2k1D} that $\MQ$ represents 
$t$ over $\ztkz$ iff $\ordt(t)-\ordt(\MQ) \geq 0$, is even and 
$\copt(t)\equiv\copt(\MQ)\bmod{\min\{8,2^{k-\ordt(t)}\}}$. But then,
\begin{align*}
x^2\MQ \equiv &t \modtk \iff 
x^2 \equiv 2^{\ordt(t)-\ordt(\MQ)}\copt(t)\copt(\MQ)^{-1} 
\bmod{2^{k-\ordt(\MQ)}} \\
&\iff x=2^{\frac{\ordt(t)-\ordt(\MQ)}{2}}y, 
\qquad y \in \bbZ/2^{k-\frac{\ordt(t)-\ordt(\MQ)}{2}}\bbZ, \\
&\qquad y^2 \equiv \copt(t)\copt(\MQ)^{-1}  \pmod{2^{k-\ordt(t)}}\;.
\end{align*}
The number of possible representations is the number of 
$y \in \bbZ/2^{k-\frac{\ordt(t)-\ordt(\MQ)}{2}}\bbZ$ satisfying
$y^2 \equiv \copt(t)\copt(\MQ)^{-1} \bmod{2^{k-\ordt(t)}}$. By Lemma
\ref{lem:Hensel2}, there are
exactly four possible $y$ over the ring
$\bbZ/2^{k-\ordp(t)}\bbZ$ if $k-\ordt(t)>2$, and 
$k-\ordt(t)$ otherwise.
Recall, the $2$-expansion of $y \in \ztkz$ 
(definition \ref{def:pexpansion}). There are 
$k-\frac{\ordt(t)-\ordt(\MQ)}{2}$ digits in the $2$-expansion;
the first $(k-\ordt(t))$ of which must be a solution to
$y^2\equiv\copt(t)\copt(\MQ)^{-1}$ modulo $2^{k-\ordt(t)}$.
The rest of the $\frac{\ordt(t)+\ordt(\MQ)}{2}$ digits
can be chosen freely from $\ztz$.
Thus, the number of representations of $t$ by $\MQ$ over $\ztkz$ is
$2^{2+\frac{\ordp(\MQ)+\ordt(t)}{2}}$ if $k-\ordt(t)>2$ and 
$(k-\ordt(t))2^{\frac{\ordt(\MQ)+\ordt(t)}{2}}$ otherwise. 
The correctness for $\ordt(t)<k$ now follows from the fact that
primitive representations exist iff $\ordp(t)=\ordp(\MQ)$.

The computation of $\ordt(\MQ)$ and $\ordt(t)$ by fast exponentiation
takes $O(\log k)$ ring operations.
\end{proof}

\subsection{Type II, $p=2$}\label{sec:Dim2II}

Recall, Definition \ref{def:BlockDiagonal}, of a type II quadratic 
form. In this section, we solve the representation problem for Type II matrices
over $\ztkz$. But first we define a scaled version of a type II matrix.

\begin{definition}\label{def:TypeIIStar}
A two-by-two matrix of the following form is called \typet\, matrix.
\[
\begin{pmatrix}a & b/2 \\ b/2 & c\end{pmatrix} \qquad 
a,b,c \in \bbZ, b \text{ odd}
\]
\end{definition}

Additionally, in this 
section we will think of \typet\,
as the following quadratic form in formal variables 
$x_1, x_2$ which take values in the ring $\ztkz$.
\begin{equation}\label{def:qftype2}
ax_1^2+bx_1x_2+cx_2^2 \qquad a,b,c \in \bbZ, b \text{ odd}\;.
\end{equation}


\begin{lemma}\label{lem:RepresentTTypeII}
Let $\MQ^*=(a,b,c)$, $b$ odd be a \typet\; integral 
quadratic form, and $t, k$ be positive integers.
If $a_1,a_2 \in \ztz$ be such that $(a_1,a_2)$ represent 
$t$ over $\ztz$ and either $a_1$ or $a_2$ is odd then
there are exactly $2^{k-1}$ distinct representations 
$(x_1,x_2)$ of $t$ over $\ztkz$ such that
$x_1\equiv a_1 \pmod {2}, x_2 \equiv a_2 \pmod {2}$.
\end{lemma}
\begin{proof}
We prove this by induction on $k$.
We show that given a representation $y_1,y_2$ of $t$ over 
the ring $\bbZ/2^i\bbZ$, for $i \geq 1$, such that at least one of
$y_1,y_2$ is odd there are exactly two representations $z_1,z_2$ of $t$ over
the ring $\bbZ/2^{i+1}\bbZ$ such that 
$z_1\equiv y_1 \pmod{2^i}, z_2\equiv y_2\ \pmod{2^i}$.

Let $(y_1,y_2)$ be a representation of $t$ by $\MQ^*$ over $\bbZ/2^i\bbZ$. 
Then, the pair of integers $(z_1, z_2)$ such that 
$(z_1,z_2) \equiv (y_1, y_2) \pmod{2^i}$ is a representation of 
$t$ over $\bbZ/2^{i+1}\bbZ$ iff 
\begin{equation}
\begin{array}{ll}
z_1 \equiv y_1 + b_1\cdot 2^{i} \pmod{2^{i+1}} 
	& z_2 \equiv y_2 + b_2 \cdot 2^{i} \pmod{2^{i+1}} \\
b_1, b_2 \in \{0,1\} & az_1^2+bz_1z_2+cz_2^2\equiv t \pmod{2^{i+1}}
\end{array}
\end{equation}

Plugging in the values of $z_1$ and $z_2$ and re-arranging we get the 
following equation.
\begin{align}\label{eq2:RepresentTTypeII}
(bb_2y_1+bb_1y_2)2^i 
  &\equiv t-(ay_1^2+by_1y_2+cy_2^2) \pmod{2^{i+1}}
\end{align}
As $b$ is odd, $b$ is invertible over $\bbZ/2^{i+1}\bbZ$. 
By assumption, $y_1,y_2$ represent $t$ over $\bbZ/2^{i}\bbZ$ and hence
$2^{i}$ divides $t-(ay_1^2+by_1y_2+cy_2^2)$. The Equation 
\ref{eq2:RepresentTTypeII} reduces to the following equation.
\begin{align}\label{eq3:RepresentTTypeII}
b_2y_1+b_1y_2 
  &\equiv \frac{t-(ay_1^2+by_1y_2+cy_2^2)}{2^ib} \pmod 2
\end{align}

We now split the proof in two cases: i) when $y_1$ is odd,
and ii) when $y_1$ is even and $y_2$ is odd.

\begin{description}
\item [$\mathbf{y_1}$ \bf{odd}.] For each choice of $b_1 \in \{0,1\}$
there is a unique choice for $b_2$ because $y_1\equiv 1 \pmod{2}$.
\begin{align*}
b_1 \in \{0,1\} \qquad b_2 &= \frac{t-(ay_1^2+by_1y_2+cy_2^2)}{2^ib}
	-b_1y_2\pmod 2
\end{align*}
\item [$\mathbf{y_1}$ \bf{even}.] In this case, $y_2 \equiv 1 \pmod{2}$ and
so $b_2$ can be chosen freely. 
\begin{align*}
b_2 \in \{0,1\} \qquad b_1 = \frac{t-(ay_1^2+by_1y_2+cy_2^2)}{2^ib} \pmod 2
\end{align*}
\end{description}
\end{proof}

\begin{lemma}\label{lem:CountTypeII}
Let $\MQ^{n=2}$ be a type II matrix and $t, k$ be positive integers. Algorithm 
\ref{alg:CountTypeII} performs $O(k\log k)$ ring operations and 
counts the number of 
primitive and non-primitive representations of $t$ by $\MQ$ over $\ztkz$.
\end{lemma}
\begin{proof}
We want to count the number of primitive and non-primitive 
$\Vx=(x_1,x_2)\in (\ztkz)^2$ such that $\Vx'\MQ\Vx\equiv t\modtk$. Let 
$\MQ=\begin{pmatrix}2^{\ell+1}a & 2^\ell b \\ 2^\ell b & 2^{\ell+1}c\end{pmatrix}$,
$b$ odd. Then,
\begin{align}\label{eq1:CountTypeII}
\Vx'\MQ\Vx \equiv t \modtk \iff 
2^{\ell+1}(ax_1^2+bx_1x_2+cx_2^2) \equiv t\modtk
\end{align}
Recall, $(x_1,x_2) \in (\ztkz)^2$ is non-primitive (definition \ref{def:Prim})
iff both $x_1$ and $x_2$ are even. We distinguish the following cases.

\paragraph{$\mathbf{\ell+1 \geq k}$.} In this case, $\MQ$ is 
identically $0$ over $\ztkz$ and hence only represents $0$. If $t$
is also $0$ over $\ztkz$ i.e.,
$\ordt(t)\geq k$ then the number of primitive and non-primitive
representations of $t$ over $\ztkz$ is $3\cdot 4^{k-1}$ and $4^{k-1}$
respectively.

\paragraph{$\mathbf{\ell + 1 > \ordt(t)}$.} Everything $\MQ$ represents
is a multiple of $2^{\ell+1}$ and so $\MQ$ cannot represent $t$ in this
case.

\paragraph{$\mathbf{\ordt(t) \geq \ell+1}$.} In this case, we divide 
Equation \ref{eq1:CountTypeII} by $2^{\ell+1}$.
\begin{align}\label{eq2:CountTypeII}
&2^{\ell+1}(ax_1^2+bx_1x_2+cx_2^2) \equiv t\modtk \nonumber\\
\iff & 
ax_1^2+bx_1x_2+cx_2^2 \equiv 2^{\ordt(t)-\ell-1}\copt(t) \pmod{2^{k-\ell-1}}
\end{align}
Both $x_1$ and $x_2$ are elements of the ring $\ztkz$. But the Equation 
\ref{eq2:CountTypeII} is defined modulo $2^{k-\ell-1}$. 
Recall, definition \ref{def:pexpansion} of $2$-expansion.
From the equivalence relation $(x \bmod q)\cdot(y \bmod q)\equiv xy \pmod{q}$, 
it follows that the last $2^{\ell+1}$ digits of both $x_1$ and $x_2$ can be
chosen freely. The number of primitive (or non-primitive) representations
of $t$ by $\MQ$ over $\ztkz$ is equal to $4^{\ell+1}$ the number of 
primitive (resp., non-primitive) solution of the following equation
over $\bbZ/2^{k-\ell-1}\bbZ$.
\begin{align}\label{eq3:CountTypeII}
ay_1^2+by_1y_2+cy_2^2 \equiv 2^{\ordt(t)-\ell-1}\copt(t) \pmod{2^{k-\ell-1}}\;.
\end{align}

\begin{algorithm}
\SetKw{Return}{return}
\SetKwInput{Input}{input}
\SetKwData{nprim}{non-prim}
\SetKwData{prim}{prim}
\SetKw{Ensure}{ensure}
\caption{\textsc{CountTypeII}($\MQ^{n=2},\symtk(t),k$)}\label{alg:CountTypeII}
$\ell:=\ordt(\MQ_{12})$\;
\uIf{$\ell+1 \geq k$}{
	\lIf{$\ordt(t)\geq k$}{$\nprim:=4^{k-1}; \prim:=4^k-\nprim$}
	\lElse{$\prim:=0; \nprim:=0$}
}
\lElseIf{$\ordp(t)<\ell+1$}{$\prim:= 0; \nprim:= 0$}
\Else{
	$\prim:= 0; \nprim:= 0; k:=k-\ell-1; t:=t/2^{\ell+1}, \MQ^*:=\MQ/2^{\ell+1}$\;
	\For{$[a_1,a_2] \in \{[0,1], [1,0], [1,1]\}$}{
		\lIf{$(a_1,a_2)\MQ^{*}(a_1,a_2)'\equiv t \pmod{2}$}{
			$\prim+=2^{k-1}$
		}
	}
	\If{$t \equiv 0 \pmod 4$}{
		\lIf{$k=1$}{$\nprim+=1$}
		\Else{
			$[n_1,n_2]:=$CountTypeII$(\MQ^*,t/4,k-2)$\;
			$\nprim+=\nprim+4(n_1+n_2)$\;
		}
	}
	$[\prim,\nprim]:=[4^{\ell+1}\cdot\prim,4^{\ell+1}\cdot\nprim]$
}
\Return{$[\prim,\nprim]$}
\end{algorithm}

Every solution $(y_1,y_2)$ of Equation \ref{eq3:CountTypeII} falls 
in two categories, i) at least one of $y_1,y_2$ is odd, or ii) both are even.
\begin{description}
\item [$\mathbf{y_1\text{ or } y_2 \text{ odd}}$] In this case, $(a_1,a_2)
= (y_1,y_2) \pmod{2}$ represents $t$ over $\ztz$ and we can apply Lemma
\ref{lem:RepresentTTypeII}. By construction, these are primitive solutions.
\item [$\mathbf{y_1\text{ and } y_2 \text{ even}}$] In this case, $4$ must
divide $t$ and by definition every representation is non-primitive. We can
divide the equation by $4$ to get the following equation (assume
$t^*=t/2^{\ell+1}, k^*=k-\ell-1$).
\begin{equation}\label{eq3:CountTypeII}
a(y_1/2)^2+b(y_1/2)(y_2/2)+c(y_2/2)^2 \equiv t^*/4 \pmod {2^{k^*-2}}
\end{equation}
\end{description}
Again, $y_1/2, y_2/2$ are elements of $\bbZ/2^{k^*-1}\bbZ$ but the Equation 
\ref{eq3:CountTypeII} is defined modulo $2^{k^*-2}$. Thus, the
last bit of $y_1$ and $y_2$ can be chosen freely. So, we can solve 
Equation \ref{eq4:CountTypeII} over $\bbZ/2^{k^*-2}\bbZ$ and then multiply
by $4$ to get the number of solutions of Equation \ref{eq3:CountTypeII}
over $\bbZ/2^{k^*}\bbZ$.
This completes the proof of correctness of 
Algorithm \ref{alg:CountTypeII}.
\begin{equation}\label{eq4:CountTypeII}
az_1^2+bz_1z_2+cz_2^2 \equiv t^*/4 \pmod {2^{k^*-2}}
\end{equation}

The algorithm works with $k$ bit numbers and each recursive call reduces 
$k$ by $2$. One can compute $\ordt(t)$ as well as $\ordt(\MQ)$ once, costing
$O(\log k)$ ring operations. Thus, the algorithm takes $O(k\log k)$ ring
operations in total.
\end{proof}

\subsection{Calculating Split Size}\label{sec:SplitSize}

Let $p$ be a prime, $k$ be a positive integer and 
$\sym$ be a $p^k$-symbol. For notational convenience, we define
$\ypk(\sym)=\{x \in \zpkz\mid \sympk(x)=\sym\}$ and $\fypk(\sym)$
as the cardinality of $\ypk(\sym)$. 
Let $\sym,\sym_1,\sym_2$ be $p^k$-symbols. In this section, we 
compute $\fspksym(\sym_1,\sym_2)$ i.e., for a fixed $t \in \ypk(\sym)$
the cardinality of the following set.
\[
\left\{(a,b) \in (\zpkz)^2 \mid \sympk(a)=\sym_1, \sympk(b)=\sym_2, 
a+b\equiv t \bmod{p^k}\right\}
\]


But first, we mention the following useful result
from \cite{Per52}. For completeness, a proof is provided in the Appendix.
\begin{lemma}\label{lem:SizeModP}
For an odd prime $p$, and non-zero $a\in \zpz$ the number 
of tuples $(x,x+a) \in (\zpz)^2$ such that 
$\legendre{x}{p}=s_1$, $\legendre{x+a}{p}=s_2$ and 
$s_1,s_2 \in \{-1,1\}$
is given by the following formula.
\begin{equation}\label{eq:ordaeqordt}
\frac{1}{4} \cdot \left\{ p- (p \bmod 4) - 
\left(\legendre{a}{p}+s_1\right) \cdot \left(\legendre{-a}{p}+s_2\right)
\right\}
\end{equation}
\end{lemma}
\bigskip

\begin{lemma}\label{lem:ORD=}
Let $p$ be a prime, $k$ be a positive integer, and $\sym_1,\sym_2, \sym$
be $p^k$-symbols with $\ordp(\sym)=\ordp(\sym_1)=\ordp(\sym_2)<k$. Then,
the size of the $\fspksym(\gamma_1,\gamma_2)$ is
$0$ for $p=2$ and otherwise it is 
calculated by substituting $\legendre{a}{p}=\sgnp(\sym_1)$,
$s_1=\sgnp(\sym_2)$ and $s_2=\sgnp(\sym)$ in 
Equation \ref{eq:ordaeqordt} and multiplying the result by $p^{k-\ordp(\sym)-1}$. 
\end{lemma}
\begin{proof}
The equality $\ordp(\sym)=\ordp(\gamma_1)=\ordp(\gamma_2)$
is not possible in case $p=2$ because the sum of two numbers of
the same $2$-order is always a number of higher $2$-order. For odd  
prime $p$, if $t \in \zpkz$ be such that $\sympk(t)=\sym$, then 
we are looking for number of solutions in $\zpkz$ of the
following equation.
\begin{align}
p^{\ordp(\sym)}\copp(a) + p^{\ordp(\sym)}\copp(b) \equiv 
p^{\ordp(\sym)}\copp(t) 
\bmod{p^k}\nonumber\\
\iff \copp(a) + \copp(b) \equiv \copp(\sym) \bmod{p^{k-\ordp(\sym)}} 
\label{EQ:ORD=}
\end{align}
The number of solutions of Equation \ref{EQ:ORD=} modulo $p$ is given 
by Lemma \ref{lem:SizeModP}. The other $(k-\ordp(\sym)-1)$ digits in the
$p$-expansion of $\copp(a)$ can be chosen freely. Thus, the number of
possibilities multiply by $p^{k-\ordp(\sym)-1}$.
\end{proof}

It turns out that if $\ordp(\sym) \neq \ordp(\sym_1)$ then 
for every $a \in \ypk(\sym_1), t \in \ypk(\sym)$ 
the value of $\sympk(t-a)$ does not depend on the specific choice
of $a$, or $t$. The following lemma proves this assertion.


\begin{lemma}\label{lem:ORDn=}
Let $p$ be a prime, $k$ be a positive integer, $\sym_1,\sym_2, \sym$
be $p^k$-symbols with $\ordp(\sym)\neq\ordp(\sym_1)$. Then, for
\begin{align*}
\sym_3 &= \left\{
	\begin{array}{ll}
	\sym & \text{if $p\neq 2$, } \ordp(\sym)<\ordp(\sym_1)\\
	\left(\ordp(\sym_1), \legendre{-1}{p}\sgnp(\sym_1)\right) 
		& \text{if $p\neq 2$, } \ordp(\sym)>\ordp(\sym_1)\\
	\left(\ordt(\sym), s_1\right) & \text{if p=2, } 
	\ordt(\sym)<\ordt(\sym_1)\text{ and, }\\
	\left(\ordt(\sym_1), s_2\right) & \text{if p=2, } 
	\ordt(\sym)>\ordt(\sym_1)\text{, where, }\\
	\end{array}\right. \\
	s_1 &:= \sgnt(\sym)-2^{\ordt(\sym_1)-\ordt(\sym)}\sgnt(\sym_1) \bmod 8\\
	s_2 &:= 2^{\ordt(\sym)-\ordt(\sym_1)}\sgnt(\sym)-\sgnt(\sym_1) \bmod 8,
\end{align*}
we have,
\[
\fspksym(\sym_1,\sym_2) = \left\{\begin{array}{ll}
\fypk(\sym_1) & \text{if }\sym_2=\sym_3, \text{ and}\\
0 & \text{otherwise.}
\end{array}\right.
\]
\end{lemma}
\begin{proof}
By the statement of the lemma,
$\ordp(\sym) \neq \ordp(\sym_1)$. 
Let $a \in \ypk(\sym_1)$ and $t \in \ypk(\sym)$ be arbitrary
elements. Then, it
suffices to show that $\sympk(t-a)=\sym_3$.

By definition of $p$-order,
$\ordp(t-a)=\min\{\ordp(a), \ordp(t)\}$. This shows that
$\ordp(t-a)=\ordp(\sym_3)$. Next, we show that 
$\sgnp(t-a)=\sgnp(\sym_3)$. We divide the proof of this fact in
two parts, depending on the prime $p$.
\begin{description}
\item[$p$ odd.] By definition of $p$-sign,
it follows that $\sgnp(t-a)=\sgnp(\sym_3)$ from the
equation below.
\[
\sgnp(t-a)=\legendre{\copp(t-a)}{p}=\left\{\begin{array}{ll}
\legendre{-\copp(a)}{p} & \text{if }\ordp(t)>\ordp(a), \text{ and,} \\
\legendre{\copp(t)}{p} & \text{otherwise.}
\end{array}\right.
\] 
\item[$p$=2.] By definition of $2$-sign, it follows that
$\sgnt(t-a)=\sgnt(\sym_3)$ from the equality below.
\[
\sgnt(t-a) = \copt(t-a) \bmod 8 = \left\{\begin{array}{ll}
s_1 & \text{if }\ordt(\sym) < \ordt(\sym_1) \text{ and,}\\
s_2 & \text{otherwise.}
\end{array}\right.
\]
\end{description}
\end{proof}

Next, we compute the cardinality of $\ypk(\sym)$ for an 
arbitrary $p^k$-symbol $\sym$.

\begin{lemma}\label{lem:SymbolSize}
Let $p$ be a prime, $k$ be a positive integer and $a \in \zpkz$ 
be a non-zero integer. Then, 
\begin{displaymath}
\fypk(\sympk(a))=\left\{ \begin{array}{ll}
\max\{2^{k-\ordt(a)-3},1\} & \text{if $p=2$}\\
\frac{p-1}{2}p^{k-\ordp(a)-1} & \text{otherwise.}
\end{array}\right.
\end{displaymath}
\end{lemma}
\begin{proof}
Let $x\in\zpkz$ be an element with the same $p$-symbol as $a$.
Then, $\ordp(x)=\ordp(a)$ and $\sgnp(x)=\sgnp(t)$. Recall
the $p$-expansion of $x$ i.e., definition \ref{def:pexpansion}.
There are $k$ digits in the $p$-expansion of $x$ for $x\in\zpkz$;
first $\ordp(a)$ of which must be identically~0.

For odd prime $p$, $\sgnp(x)=\sgnp(a)$ iff 
$\legendre{\copp(x)\copp(t)}{p}=1$. Thus, the $(\ordp(a)+1)$'th
digit of $x$ must be a non-zero element of $\zpz$ with the same
sign as $\legendre{\copp(a)}{p}$. By Lemma \ref{lem:QR}, there
are $\frac{p-1}{2}$ possibilities for the $(\ordp(a)+1)$'th
digit of $x$. The rest can be chosen freely from $\zpz$.

For the prime~2, $\sgnt(x)=\sgnt(a)$ iff 
$\copt(x)\equiv\copt(a) \bmod 8$. Thus, the digits 
$(\ordp(a)+1),\cdots,(\ordp(a)+2)$ of $x$ must match those
of $a$. The rest can be chosen freely from $\ztz$.
\end{proof}

We are now ready to compute the size of the class 
$\fspksym(\sym_1,\sym_2)$.

\begin{lemma}\label{lem:SplitClassSize}
Let $p$ be a prime, $k$ be an integer and $\sym, \sym_1, \sym_2$
be $p^k$-symbols. Then, Algorithm~\label{alg:Split} computes 
$\fspksym(\sym_1,\sym_2)$ and performs only $O(1)$ operations
over integers of size $O(p^{2k})$.
\end{lemma}

\begin{proof}
Recall the definition of $\fspksym(\sym_1,\sym_2)$. For any choice
$t \in \ypk(\sym)$, it is the size of the following set.
\begin{equation}\label{def:splitset}
\{(a,b) \mid a \in \ypk(\sym_1), b \in \ypk(\sym_2), a+b \equiv t \bmod{p^k}\}
\end{equation}

Consider the situation when $\ordp(\sym)=\infty$. In this case,
$t\equiv 0\bmod{p^k}$. By Equation \ref{def:splitset}, it follows
that $\ordp(\sym_1)=\ordp(\sym_2)$. There are two possible sub-cases.
\begin{description}
\item[$\ordp(\sym_1)=\infty$] In this case, $\ordp(\sym_2)=\infty$
and the only possible element in the set $\fspksym(\sym_1,\sym_2)$ is
$(0,0)$.
\item[$\ordp(\sym_1)<k$] Note that $\ordp(\sym_1)=\ordp(\sym_2)(=\alpha)$, say. 
Then, we are looking for the number of solutions of the following equation.
\[
p^\alpha \nu_1 + p^\alpha \nu_2 \equiv 0 \bmod{p^k}\qquad
\nu_1, \nu_2 \in (\zpkz)^\times 
\]
The number of solutions is equal to the number of possible elements in
$\zpkz$ with $p$-order $\alpha$. This equals $p^{k-\alpha-1}$ times $(p-1)$.
\end{description}

Otherwise, $\ordp(\sym)<k$. 
If exactly one of $\ordp(\sym_1)$ and $\ordp(\sym_2)$
equals $\infty$ then $\fspksym(\sym_1,\sym_2)=1$. If both equal $\infty$,
then the set $\fspksym(\sym_1,\sym_2)$ is empty.

It remains to consider the case when 
$\ordp(\sym),\ordp(\sym_1),\ordp(\sym_2)<k$. There are three sub-cases.
\begin{description}
\item[$\ordp(\sym) \neq \ordp(\sym_1)$] The correctness follows from
Lemma \ref{lem:ORDn=}.
\item[$\ordp(\sym) \neq \ordp(\sym_2)$] The correctness follows from
Lemma \ref{lem:ORDn=}.
\item[$\ordp(\sym) = \ordp(\sym_1)=\ordp(\sym_2)$] The correctness follows from
Lemma \ref{lem:ORD=}.
\end{description}

The size of the set $\fspksym(\sym_1,\sym_2)$ is bounded by the number
of elements in $(\zpkz)^2$, which is $p^{2k}$.

\begin{algorithm}
\SetKw{Return}{return}
\SetKw{And}{and}
\SetKw{Or}{or}
\caption{\textsc{SplitClassSize}($p,k,\sym, \sym_1,\sym_2$)}\label{alg:Split}
\uIf{$\ordp(\sym)=\infty$}{
	\lIf{$\ordp(\sym_1) \neq \ordp(\sym_2)$}{\Return $0$}
	\lElseIf{$\ordp(\sym_1)=\ordp(\sym_2)=\infty$}{\Return $1$}
	\lElse{\Return $(p-1)p^{k-\ordp(\sym_1)-1}$}
}
\lIf{$\ordp(\sym_1)=\infty$ \And $\ordp(\sym_2)=\infty$}{\Return $0$}
\lIf{$\ordp(\sym_1)=\infty$ \Or $\ordp(\sym_2)=\infty$}{\Return $1$}
\lIf{$\ordp(\sym)=\ordp(\sym_1)$}{\textsc{swap}$(\sym_1,\sym_2)$}
\uIf{$\ordp(\sym)=\ordp(\sym_1)$}{
	\lIf{$p=2$}{\Return $0$}
	\Else{
	$x :=(\sgnp(\sym_2)+\sgnp(\sym_1))
	\left(\legendre{-1}{p}\sgnp(\sym_1)+\sgnp(\sym)\right)$\;
		\Return $\frac{p^{k-\ordp(\sym)-1}}{4}(p-(p\bmod 4)-x)$
	}
}
\Else{
	\uIf{$p=2$}{
		\lIf{$\ordp(\sym_1)<\ordp(\sym)$}{
			$\sym_3:=(\ordp(\sym_1), 
			  2^{\ordp(\sym)-\ordp(\sym_1)}\sgnt(\sym)-\sgnt(\sym_1) \bmod 8)$
		}
		\lElse{
			$\sym_3:=(\ordp(\sym), 
			  \sgnt(\sym)-2^{\ordp(\sym_1)-\ordp(\sym)}\sgnt(\sym_1) \bmod 8)$
		}
	}
	\Else{
		\lIf{$\ordp(\sym_1)<\ordp(\sym)$}{
			$\sym_3:=(\ordp(\sym_1), \legendre{-1}{p}\sgnp(\sym_1))$
		}
		\lElse{
			$\sym_3:=\sym$
		}
	}
	\lIf{$\sym_3 = \sym_2$}{\Return $\fypk(\sym_1)$}
	\lElse{\Return $0$}
}
\end{algorithm}

\end{proof}

\subsection{Dimension $>1$, Any Prime}\label{sec:Count}

\begin{theorem}\label{thm:RepresentTbyQModPk}
Let $p$ be a prime, $k$ be a positive integer,
$t \in \zpkz$ and $\MQ^n$ 
be an integral quadratic form. Then, there is an
algorithm that counts the total, primitive and non-primitive
number of $p^k$-representations of $t$ by $\MQ$ in
$O(n^{1+\omega}\log k+nk^3+n\log p)$ ring operations over $\zpkz$.
\end{theorem}
\begin{proof}
Let $\MD^n=\diag(\MD_1,\cdots,\MD_a)$ be the block diagonal quadratic form
(see Definition \ref{def:BlockDiagonal}) 
equivalent to $\MQ$ over the ring $\zpkz$, where each $\MD_i$ is
either a Type I or a Type II block (takes $O(n^{1+\omega}\log k)$ ring
operations over $\zpkz$ to block diagonalize).
From Theorem \ref{thm:BlockDiagonal}, Lemma 
\ref{lem:QFEquivalence} and Lemma \ref{lem:SymbolInvariance}, 
it follows that 
\begin{gather*}
\fapk(\MQ,t)=\fapk(\MD,\sympk(t)) \\
\fbpk(\MQ,t)=\fbpk(\MD,\sympk(t)) \\
\fcpk(\MQ,t)=\fapk(\MD,\sympk(t))
\end{gather*}

We show how to calculate $\fapk(\MD,\sympk(t))$, using dynamic 
programming. Let us define $\MQ_i$, $\ORD$ and $\SGN$ as
follows.
\begin{align*}
\MQ_i&=\diag(\MD_1,\dots,\MD_i) \\
\ORD&=\{\infty, 0, \cdots, k-1\} \\
\SGN &= \left\{\begin{array}{ll}
\{0,1,3,5,7\} & \text{if $p=2$} \\
\{0,+1,-1\} & \text{Otherwise} 
\end{array}\right.
\end{align*}

One now proceeds as follows.
\begin{enumerate}[(i.)]
\item For each possible $p^k$-symbol 
$\sym \in \ORD\times\SGN$ and each $i \in [a]$;
compute $\fapk(\MD_i,\sym)$. 
The calculation can be done 
using Lemma 
\ref{lem:CountModPk1D}, Lemma \ref{lem:CountMod2k1D} or 
Lemma \ref{lem:CountTypeII} depending on
$p$ and the type of $\MD_i$. There are $O(k)$ possible
values for $\sym$, $a$ is bounded by $n$ and the 
calculation of $\fapk(\MD_i,\sym)$ takes $O(k\log k+\log p)$
ring operations. Thus, the number of ring operations needed 
for this step is $O(nk\log k + n \log p)$.
\item For each possible $p^k$-symbol triples 
$\sym, \sym_1, \sym_2 \in \ORD\times\SGN$ compute the
split set size $\fspksym(\sym_1,\sym_2)$ by Lemma 
\ref{lem:ORD=} and Lemma \ref{lem:ORDn=} taking $O(1)$ each. In total, this
step requires $O(k^3)$ operations over integers.
\item Starting from $i=1$, for each possible value of $p^k$-symbol
$\sym$, use the following formula (i.e.,
dynamic programming) to build the final result i.e., 
$\fapk(\MD, \SYM_{p^k}(t))$.
\begin{align}\label{EQ:dynamic}
\fapk(\MQ_i\oplus\MD_{i+1},\gamma) = 
 \sum_{\gamma_1, \gamma_2 \in \ORD\times\SGN}
\fS^{\gamma}_{p^k}(\gamma_1,\gamma_2)\cdot\fapk(\MQ_i,\gamma_1)
	\cdot\fapk(\MD_{i+1},\gamma_2) 
\end{align}
Note that $\fspksym(\sym_1,\sym_2)$ and $\fapk(\MD_{i+1}, \sym_2)$
have been pre-computed. For each $i$, the summation in Equation
\ref{EQ:dynamic} takes $O(k^2)$ operations over integers. As the
computation needs to be done for each possible $p^k$-symbol $\sym$,
in total, this step takes $O(k^3)$ operations over $\bbZ$. The final
solution can then be built by dynamic programming in $O(nk^3)$ operations
over the integers.
\end{enumerate}

The overall complexity of this algorithm is 
$O(n^{1+\omega}\log k+nk^3+n\log p)$.

The calculation for $\fcpk(\MQ,t)$ is the same i.e., replace $\fapk$ by
$\fcpk$. The number 
$\fbpk(\MQ,t)$ can be calculated by taking the difference
of $\fapk(\MQ,t)$ and $\fcpk(\MQ,t)$.
\end{proof}


\subsection{Computing Local Density}\label{Sec:Density}

Let $\MQ^n$ be a quadratic form, $k, t$ be positive integers and $p$ be a prime.
As defined earlier, $\fapk(\MQ,t)$ is the number of solutions of 
$\Vx'\MQ\Vx\equiv t \bmod{p^k}$ over the ring $\zpkz$. For sufficiently large
$k$, the quantity $\fapk(\MQ,t)=\alpha_p(\MQ,t)p^{k(n-1)}$, where $\alpha_p(\MQ,t)$
is a function $\MQ$ and $t$ and is called the local density. Note that
there are $p^{kn}$ possible choices for $\Vx \in (\zpkz)^n$ and so one can
interpret $\alpha_p(\MQ,t)/p^k$ as the probability that a random choice from
$(\zpkz)^n$ will satisfy the equation $\Vx'\MQ\Vx\equiv t\bmod{p^k}$.

There are several papers on computing the local density \cite{Yang98}. Our 
methods give an alternative way to compute the local density in polynomial
time. It can be shown that $k=1+\ordp(8t\det(\MQ))$ suffices for computing the
local density (pp~378-381, \cite{CS99}). Thus, 
\[
\alpha_p(\MQ,t) = \frac{\fA_{p^s}(\MQ,t)}{p^{s(n-1)}} \qquad s = 1+\ordp(8t\det(\MQ))
\]
This implies that the computation of $\fapk(\MQ,t)$ for 
$k\geq 1+\ordp(8t\det(\MQ))$ can be done in number of
ring operations over $\zpkz$ that does not depend on $k$. Note that
the bit complexity of the algorithm will still depend on $k$ 
because the ring operations are performed in the ring $\zpkz$.

We are not aware if a similar idea works for $\fcpk(\MQ,t)$ or
$\fbpk(\MQ,t)$.

\section{Sampling a Uniform Representation}\label{Sec:Sample}

Let $\MQ^n$ be an integral quadratic form, $p$ be a prime, 
$k$ be a positive integer and $t$ be an element of the ring $\zpkz$. 
In this section, we generate a uniformly random primitive representation 
of $t$ by $\MQ$ over $\zpkz$. The
algorithm runs in time $\poly(n,k,\log p)$ and 
fails with constant probability.
Otherwise, the algorithm outputs a 
uniformly random primitive representation.
A uniform representation and a uniform non-primitive representation can 
also be generated similarly. Note that we are assuming that a 
representation of the correct kind (primitive, non-primitive) 
exists. 

\subsection{Sampling Uniformly from a Split}
\label{sec:SampleFromSplit}

Let $p$ be a prime, $k$ be a positive integer, $t$ be 
an element of $\zpkz$ and $\gamma_1, \gamma_2 \in \ORD\times\SGN$ 
be a pair of $p^k$-symbols.
We show how to generate a uniform pair $(a,b) \in (\zpkz)^2$ from 
the following set;
\begin{align}\label{def:Set}
\Big\{(x,y) \in (\zpkz)^2\mid \sympk(x)=\gamma_1, \sympk(y)=\gamma_2, 
x+y\equiv t\modpk\Big\}\;
\end{align}

Recall Lemma \ref{lem:ORD=} and Lemma \ref{lem:ORDn=}. There are 
three possible cases and in
each case a uniform pair can be generated as follows.
\begin{description}
\item[$\ordp(\gamma_1) \neq \ordp(t)$] In this case, by Lemma 
\ref{lem:ORDn=}, a
uniform pair can be generated by picking a uniform $a$ from 
$\spk(\gamma_1)$ (use Lemma \ref{lem:SampleFromSymbol}), 
and outputting $(a, t-a \bmod{p^k})$.
\item[$\ordp(\gamma_1) = \ordp(t) \neq \ordp(\gamma_2)$] Pick a 
uniform $a$ from $\spk(\gamma_2)$ (use Lemma \ref{lem:SampleFromSymbol}),
 and output $(t-a\bmod{p^k},a)$. Correctness follows from Lemma 
\ref{lem:ORDn=}.
\item[$\ordp(\gamma_1)=\ordp(\gamma_2)=\ordp(t)$] 
Recall Lemma \ref{lem:ORD=}. This can never 
happen when $p=2$. Otherwise, generate using 
Lemma \ref{lem:SampleModP}. 
\end{description}

\begin{lemma}\label{lem:SampleFromSymbol}
Let $p$ be a prime, $k$ be a positive integer, 
and $\gamma$ be a $p^k$-symbol.
Then, there is an algorithm that performs $O(\log p)$ ring 
operations and (i) for $p=2$,
outputs a uniform element from the set 
$\ytk(\gamma)$, and (ii) for $p$ odd, 
with probability $1/6$ outputs a uniform element from
the set $\ypk(\gamma)$.
\end{lemma}
\begin{proof}
If $\sym$ is $\sympk(0)$ then output $0$. Otherwise, proceed as follows.
Let $\ordp(\gamma)=i$ and $\sgnp(\gamma)=s$.
By Definition \ref{def:SymbolOneDim}, $s$ is in the set
$\{1,-1\}$ for odd prime $p$ and $\{1,3,5,7\}$ otherwise. If 
$r$ is an element of the set $\{x\in\zpkz\mid \sympk(x)=\gamma\}$
then $r = p^{i}b$, where $1\leq b < p^{k-i}$ is a number coprime 
to $p$ and $\sgnp(b)=s$. Thus, it suffices to generate a uniform 
number in the set
\[
S=\{y\mid 1\leq y < p^{k-i}, \gcd(y,p)=1, \sgnp(y)=s\}\;.
\]

Consider the case of $p$ odd. Recall Definition \ref{def:pexpansion}. 
The numbers in the set $S$ 
are of the form $dp+\tau$, where $\tau \in \{1\leq x<p\mid \sgnp(x)=s\}$
and $d$ is an integer satisfying $0 \leq d \leq \frac{p^{k-i}-1-\tau}{p}$.
A uniform element from $S$ can be chosen by picking a uniform integer 
$d$ in the set $\{0,\cdots, \lfloor\frac{p^{k-i}-1-\tau}{p}\rfloor\}$
and picking a uniform non-zero integer from $\zpz$ with sign $s$. 
Exactly half of the non-zero elements of $\zpz$ have sign $s$ and so 
picking one at random has a success probability of exactly 
$\frac{p-1}{2p}$ (we reject otherwise and output fail).

Otherwise, $p=2$. Additionally, assume that $k-i>2$. Then, the numbers
in the set $S$ are of the form $8d+s$, where $d$ is in the set 
$\{0,\cdots, \lfloor \frac{2^{k-i}-s}{8} \rfloor\}$. If $k-i\leq 2$
then there is only one possible element in $\ztkz$ with symbol $\gamma$, 
which is $2^{k-i}s$. 

For $p$ odd, the algorithm needs to generate a uniform non-zero 
element in $\zpz$ with sign $s$ and an element from a set of size at most
$p^{k-i}$. This needs $O(\log p)$ ring operations (for 
computing the Legendre symbol) and the probability of success 
is at least $\frac{1}{6}$ because $p\geq 3$.
In case of $p=2$, the algorithm
only needs to generate an element in a set of size at most $2^{k-i}$, 
completing the proof.
\end{proof}

\begin{lemma}\label{lem:SampleModP}
Let $p$ be an odd prime, $k, i$ be positive integers, 
$t$ be an element of $\zpkz$ and $\gamma_1, \gamma_2$ 
be $p^k$-symbols
with $\ordp(t)=\ordp(\gamma_1)=\ordp(\gamma_2)=i$. Then, there is an
algorithm that performs $O(\log p)$ ring operations, 
and with probability at least $1/12$ outputs a uniform pair
from the set $S$ defined in Equation \ref{def:Set}, if it is
non-empty.
\end{lemma}
\begin{proof}
In case the set is empty, it can be detected by 
Lemma \ref{lem:ORD=}. We now assume that the set is non-empty.

Any pair $(a,b) \in (\zpkz)^2$ from the set $S$ is of the following
form
\begin{align}
a = p^i\tau_1 \qquad b &= p^i\tau_2 \qquad \tau_1, \tau_2 < p^{k-i} \\
\tau_1 + \tau_2 &\equiv \copp(t) \pmod{p^{k-i}} \\
\legendre{\tau_1}{p} = \sgnp(\gamma_1) &~~~~ \legendre{\tau_2}{p} 
= \sgnp(\gamma_2) \label{SampleModP:sign}\;.
\end{align}
The candidate algorithm is as follows.
By construction, the algorithm returns a valid and uniform pair
from the set $S$, if it succeeds. 

\begin{figure}[h]
\setlength{\unitlength}{0.14in}
\centering
\begin{picture}(30,7)
\put(0,6){$\tau_1 \leftarrow 
\big\{x\mid 1\leq x \leq p^{k-i}, \gcd(x,p)=1\big\}$}
\put(0,4){$\tau_2 = \copp(t) - \tau_1 \modpk$}
\put(0,2){\bf{if} $\legendre{\tau_1}{p}\neq
\sgnp(\gamma_1)$ \bf{or} $\legendre{\tau_2}{p}\neq\sgnp(\gamma_2)$ \bf{then: }
return $\bot$}
\put(0,0){\bf{else: }return $(p^i\tau_1, p^i\tau_2)$}
\end{picture}
\end{figure}

The probability of failure of
the algorithm can be calculated using Lemma \ref{lem:SizeModP},
as follows.
The numbers $\tau_1$ and $\tau_2$ are both elements from the set 
$\big\{x\mid 1\leq x \leq p^{k-i}, \gcd(x,p)=1\big\}$. If $\tau_1$
is uniform, then so is $\tau_2$. Recall the definition of $p$-expansion
i.e., Definition \ref{def:pexpansion}.
Both $\tau_1$ and $\tau_2$ are of the form 
$dp+a$, where $a$ is in the set $\{1,\cdots,p-1\}$ and $d$ is in the set
$\{0,\cdots, \lfloor \frac{p^{k-i}-a}{p} \rfloor\}$. For the pair
$\tau_1=d_1p+a_1$ and $\tau_2=d_2p+a_2$ to satisfy the Equation 
\ref{SampleModP:sign}, it is necessary and sufficient that 
$\legendre{a_1}{p}=\sgnp(\gamma_1)$ and $\legendre{a_2}{p}=\sgnp(\gamma_2)$.
For a randomly picked $\tau_1$, by Lemma \ref{lem:SampleModP} and Table
\ref{table}, this happens with probability at least $\frac{p-5}{4(p-1)}$.
The probability of failure, for one iteration, is at most 
$\frac{3p+1}{4(p-1)}$, which is $<\frac{11}{12}$ for $p>7$. For $p\leq 7$,
we find $a_1$ and $a_2$ by brute force.

One iteration of the algorithm performs $O(\log p)$ ring 
operations and fails with probability at most $\frac{11}{12}$. 
\end{proof}

\subsection{Sampling a Representation}\label{Sec:SamplingTheorem}

This section deals with sampling a uniform (primitive, non-primitive) 
representation of $t$ by $\MQ$ over $\zpkz$.

The base case is when $\MQ$ is a single 
block (Definition \ref{def:BlockDiagonal}).
There are three distinct possibilities. The block $\MQ$ can be a 
Type II block (with $p=2$), 
a one dimensional block (with $p=2$) or
a one dimensional block (with $p$ odd). In each case,
a uniform (primitive, non-primitive) representation can be sampled. The
complexity of sampling is the same irrespective of primitiveness of the
representation.

The following lemma samples a uniform representation when $\MQ$ is of
type II (with $p=2$).

\begin{lemma}\label{lem:SampleTypeII}
Let $\MQ$ be a type II quadratic form, $k$ be a positive 
integer and $t$ be an element of $\ztkz$. Then, there 
exists an algorithm performing $O(k^2\log k)$ ring operations which outputs
a uniform representation
of $t$ by $\MQ$ over $\ztkz$.
\end{lemma}
\begin{proof}
Let 
$\MQ=\begin{pmatrix}2^{\ell+1}a & 2^\ell b \\ 2^\ell b & 2^{\ell+1}c\end{pmatrix}$,
$b$ odd. Then,
\begin{align}\label{SampleTypeII:EQ0}
\Vx'\MQ\Vx \equiv t \modtk \iff 
2^{\ell+1}(ax_1^2+bx_1x_2+cx_2^2) \equiv t\modtk
\end{align}

We distinguish the following cases.

\paragraph{$\mathbf{\ell+1 \geq k}$.} In this case, $\MQ$ is 
identically $0$ over $\ztkz$ and hence only represents $0$. If $t$
is also $0$ then a uniform representation can be found by sampling
$x_1$ and $x_2$ independently at random from $\ztkz$.

\paragraph{$\mathbf{\ell + 1 > \ordt(t)}$.} No representations exist,
see Lemma \ref{lem:CountTypeII}.

\paragraph{$\mathbf{\ordt(t) \geq \ell+1}$.} In this case, we divide 
Equation \ref{SampleTypeII:EQ0} by $2^{\ell+1}$.
\begin{align}\label{SampleTypeII:EQ1}
&2^{\ell+1}(ax_1^2+bx_1x_2+cx_2^2) \equiv t\modtk \nonumber\\
\iff & 
ax_1^2+bx_1x_2+cx_2^2 \equiv 2^{\ordt(t)-\ell-1}\copt(t) \pmod{2^{k-\ell-1}}
\end{align}
Both $x_1$ and $x_2$ are elements of the ring $\ztkz$. But the Equation 
\ref{SampleTypeII:EQ1} is defined modulo $2^{k-\ell-1}$. 
Recall, Definition \ref{def:pexpansion} of $2$-expansion.
From the equivalence relation $(x \bmod q)\cdot(y \bmod q)\equiv xy \pmod{q}$, 
it follows that the last $2^{\ell+1}$ digits of both $x_1$ and $x_2$ can be
chosen freely. Hence, we pick them uniformly at random. The problem then reduces 
to finding a uniform solution to the following equation.
\begin{align}\label{SampleTypeII:EQ2}
ay_1^2+by_1y_2+cy_2^2 \equiv 2^{\ordt(t)-\ell-1}\copt(t) \pmod{2^{k-\ell-1}}\;.
\end{align}

Every solution $(y_1,y_2)$ of Equation \ref{SampleTypeII:EQ2} falls 
in four categories, i) $y_1$ odd, $y_2$ is odd, ii) $y_1$ even, $y_2$ odd,
(iii) $y_1$ odd, $y_2$ even, and (iv) $y_1$ even, $y_2$ even.

We calculate the number of representations of each kind. In the first three 
cases the number can be calculated by using Lemma 
\ref{lem:RepresentTTypeII}. A solution with say $y_1$ odd and $y_2$ even 
exists iff $y_1=1, y_2=0$ satisfies the Equation \ref{SampleTypeII:EQ2} 
modulo $2$. The number of solutions will be $2^{k-\ell-2}$ and $0$ otherwise.
The number of solutions in case (iv) i.e., both $y_1$ and $y_2$ are even,
is $4$ times the number of solutions to the following equation.
\begin{align}\label{SampleTypeII:EQ3}
az_1^2+bz_1z_2+cz_2^2 \equiv 2^{\ordt(t)-\ell-3}\copt(t) \pmod{2^{k-\ell-3}}\;.
\end{align}
The number of solutions of Equation \ref{SampleTypeII:EQ3} can be computed
by Lemma \ref{lem:CountTypeII}.

Once we have the number of solutions in each case, we pick a case with the
corresponding probability i.e., case (i) is picked with probability 
the number of solutions in case (i) divided by the total number of solutions
of Equation \ref{SampleTypeII:EQ2}.

In case (i), (ii) and (iii), a uniform solution can be constructed using
Lemma \ref{lem:RepresentTTypeII}, bit by bit. 
In case (iv), we fix the first bit of $y_1$ and $y_2$ to be $0$, divide
the equation by $4$ and find a uniform solution of Equation 
\ref{SampleTypeII:EQ3} over $\bbZ/2^{k-\ell-3}$. This can be done recursively.

The algorithm performs $O(k\log k)$ ring operations in counting
the number of solutions corresponding the four cases and may recursively call
itself. In case a recursive call is made, $k$ reduces by at least $3$.
Thus, the number of ring operations is $O(k^2\log k)$.
\end{proof}

From the proof of Lemma \ref{lem:SampleTypeII}, it follows that a uniform
primitive and a uniform non-primitive representation can also be sampled
in $O(k^2\log k)$ ring operations.

The following lemmas show that a uniform representation, a uniform
primitive representation and a uniform non-primitive representation
of $t$ by $\MQ$ over $\zpkz$ can also be sampled when $\MQ$ is one
dimensional.

\begin{lemma}\label{lem:SampleTypeIP}
Let $\MQ, t, k$ be integers, and $p$ be an odd prime.
Then, there exists
an algorithm that performs $O(\log k+\log p)$ ring operations,
failing with constant probability.
Otherwise, it outputs a uniform (primitive, 
non-primitive) representation of $t$ by $\MQ$ over $\zpkz$.
\end{lemma}
\begin{proof}
The proof of Lemma \ref{lem:CountModPk1D} is constructive
and with minor modifications, it can be used to generate
uniform representations.

When $\ordp(t) < k$ and the conditions in Lemma 
\ref{lem:RepresentTModPk1D} are satisfied then we find a 
square root of $\copp(t)\copp(\MQ)^{-1}$ modulo 
$p^{k-\ordp(\MQ)}$, see proof of Lemma \ref{lem:CountModPk1D}.
This can be done in 
$O(\log k+\log p)$ ring operations, 
using Lemma \ref{lem:HenselP}.
\end{proof}

\begin{lemma}\label{lem:SampleTypeI2}
Let $\MQ, t$ and $k$ be positive integers. Then, there exists
an algorithm that performs $O(k)$ ring operations and outputs 
a uniform primitive (or
non-primitive) representation of $t$ by $\MQ$ over $\ztkz$.
\end{lemma}
\begin{proof}
The proof of Lemma \ref{lem:CountMod2k1D} is constructive
and with minor modifications, it can be used to generate
uniform representations.

The square root of $\copt(t)\copt(\MQ)^{-1}$ modulo $2^k$
in Lemma \ref{lem:CountMod2k1D} can be found by Lemma
\ref{lem:Hensel2}.
\end{proof}

We now prove the main result of this paper. 

\begin{proof}(Theorem \ref{thm:Sample2}, Theorem
\ref{thm:SampleP})
The steps in the algorithm for generating a uniform primitive 
representation of $t$ by $\MQ^n$ over $\zpkz$ are as follows.

\begin{enumerate}[(i.)]
\item Find $\MU \in \gln(\zpkz)$ that block diagonalizes $\MQ$, using
Theorem \ref{thm:BlockDiagonal}. Let 
$\MD:=\MU'\MQ\MU \pmod{p^k}$ and $\MD=\MD_1 \oplus \cdots \oplus \MD_m$,
where $\MD_i, i \in [m]$ are single blocks (see Definition
\ref{def:BlockDiagonal}).

\item 
Fix the following notation, where $\MD^n=\MD_1^{n_1}\oplus\MD_{2+}^{n-n_1}$.
\begin{align*}
\MD_{2+}&=\MD_2\oplus\cdots\oplus\MD_m \\
\prim(\Vx) &= \left\{\begin{array}{ll}
1 & \text{if $\Vx$ is primitive and $\Vx'\MD\Vx\equiv t\modpk$} \\
0 & \text{Otherwise} 
\end{array}\right.
\end{align*}
Compute the total number of primitive representations of $t$ by 
$\MD$ over $\zpkz$ i.e., $\fbpk(\MD,t)$ 
(see Theorem \ref{thm:RepresentTbyQModPk}).
For every pair of $p^k$-symbols 
$\gamma_1, \gamma_2 \in \ORD \times \SGN$ calculate the following
numbers (also using Theorem \ref{thm:RepresentTbyQModPk}).
\[
\fcpk(\MD_1, \gamma_1) \qquad \fbpk(\MD_1, \gamma_1) \qquad
\fcpk(\MD_{2+}, \gamma_2) \qquad \fbpk(\MD_{2+}, \gamma_2)\;
\]
By definition of primitiveness (Definition \ref{def:Prim}), it follows that
a vector $\Vx=(\Vx_1^{n_1},\Vxtilde^{n-n_1}) \in (\zpkz)^n$ is primitive
iff at least one of $\Vx_1, \Vxtilde$ is primitive. Thus,
the probability that a random primitive representation 
$\Vx=(\Vx_1^{n_1},\Vxtilde^{n-n_1})$ of $t$ by $\MD$ satisfies
the conditions $\sympk(\Vx_1'\MD_1\Vx_1\bmod {p^k})=\gamma_1$ and
$\sympk(\Vxtilde'\MD_{2+}\Vxtilde\bmod {p^k})=\gamma_2$ can be calculated
as follows.
\begin{align}\label{eq:Probability}
&\Pr_{\prim(\Vx=(\Vx_1,\Vxtilde))}\left[\sympk(\Vx_1'\MD_1\Vx_1) 
=\gamma_1 \wedge (\sympk(\Vxtilde'\MD_{2+}\Vxtilde)=\gamma_2)
\right]=\nonumber\\
&\frac{\Big(\fcpk(\MD_1,\gamma_1)\fbpk(\MD_{2+},\gamma_2)+
\fbpk(\MD_1,\gamma_1)\fcpk(\MD_{2+},\gamma_2)+
\fbpk(\MD_1,\gamma_1)\fbpk(\MD_{2+},\gamma_2) \Big)}{\fbpk(\MD,t)}
\end{align}
\item
There are three distinct cases here: (I) $\Vx_1$ is non-primitive, 
$\Vxtilde$ is primitive,
(II) $\Vx_1$ is primitive, $\Vxtilde$ is non-primitive, and
(III) $\Vx_1, \Vxtilde$ are both primitive. The probability
of the individual cases is the corresponding summand in Equation
\ref{eq:Probability}. 
Sample one of the three cases of the summand 
in Equation \ref{eq:Probability}, with the corresponding probability.

\item 
The next step is to sample a uniform 
pair $(a,b)$ from the set $\fspkt(\sym_1,\sym_2)$ given 
in Equation \ref{def:Set}.
Recall Lemma \ref{lem:ORD=} and Lemma \ref{lem:ORDn=}. There are 
three possible cases and in
each case a uniform pair can be generated as follows.
\begin{description}
\item[$\ordp(\gamma_1) \neq \ordp(t)$] In this case, by Lemma 
\ref{lem:ORDn=}, a
uniform pair can be generated by picking a uniform $a$ from 
$\spk(\gamma_1)$ (use Lemma \ref{lem:SampleFromSymbol}), 
and outputting $(a, t-a \bmod{p^k})$.
\item[$\ordp(\gamma_1) = \ordp(t) \neq \ordp(\gamma_2)$] Pick a 
uniform $a$ from $\spk(\gamma_2)$ (use Lemma \ref{lem:SampleFromSymbol}),
 and output $(t-a\bmod{p^k},a)$. Correctness follows from Lemma 
\ref{lem:ORDn=}.
\item[$\ordp(\gamma_1)=\ordp(\gamma_2)=\ordp(t)$] 
Recall Lemma \ref{lem:ORD=}. This can never 
happen when $p=2$. Otherwise, generate using 
Lemma \ref{lem:SampleModP}. 
\end{description}

\item
Depending on the cases, do one of the following: (I) 
generate a uniform non-primitive representation of $a$ by 
$\MD_1$  and a uniform primitive representation of $b$
by $\MD_{2+}$ recursively, (II) generate a uniform primitive 
representation of $a$ by $\MD_1$ and a uniform non-primitive 
representation of $b$ by $\MD_{2+}$ recursively, and (III) generate a
uniform primitive representation of $a$ by $\MD_1$ and $b$
by $\MD_{2+}$ recursively. Let the representations be 
$\Vx_1 \in (\zpkz)^{n_1}$
and $\Vxtilde \in (\zpkz)^{n-n_1}$. Then, output 
$\MU^{-1}(\Vx_1,\Vxtilde)$. 
\end{enumerate}
By construction, this is a uniform
primitive representation of $t$ by $\MQ$ over $\zpkz$.

The block diagonalization takes 
$O(n^{\omega+1}\log k)$ ring operations (Theorem \ref{thm:BlockDiagonal}).
By Theorem \ref{thm:RepresentTbyQModPk}, the dynamic
programming approach computes the number of primitive and non-primitive
representation of every symbol modulo $p^k$ for all intermediate
diagonal forms $\MD_1\oplus \cdots \oplus \MD_i$, $i\leq m$. 
The diagonalization and solution counting need only be done
once. In total, this takes $O(n^{1+\omega}\log k + nk^3+n\log p)$ 
ring operations over $\zpkz$.

In each recursive step, we need to compute 
the probabilities of the $O(k^2)$ symbol pairs, taking $O(k^2)$ 
operations over integers. Once we have picked a $p^k$-symbol
pair $(\sym_1,\sym_2)$ for which we are going to generate a solution,
we sample $(a,b)$ from $\fspkt(\sym_1,\sym_2)$. This costs another
$O(\log p)$ ring operations (Lemma \ref{lem:SampleFromSymbol},
Lemma \ref{lem:SampleModP}). Then, we need to sample a uniform
primitive/non-primitive $p^k$-representation using a single block
(Lemma \ref{lem:SampleTypeII}, Lemma \ref{lem:SampleTypeI2}, 
Lemma \ref{lem:SampleTypeIP}). This costs $O(k^2\log k + \log p)$
at most. Finally, we need to recursively
sample a uniform primitive/non-primitive representation for
block diagonal form with strictly smaller number of blocks. This
takes a total of $O(k^2\log k + \log p)$ ring operations for one 
step and $O(nk^2\log k + n \log p)$ in total for the entire recursion.
Thus, the algorithm performs $O(n^{1+\omega}\log k+nk^3+n\log p)$
ring operations over $\zpkz$.

The sampling of a uniform non-primitive representation is relatively
easier because for $(\Vx_1, \Vxtilde)$ to be non-primitive, both
$\Vx_1$ and $\Vxtilde$ must be non-primitive. Thus, we go over all
possible $p^k$-symbol pairs $(\sym_1,\sym_2)$ and compute the 
corresponding probability as in Equation \ref{eq:Probability} as
follows.
\begin{align*}
\Pr_{\Vx \text{ non-primitive}}&\left[\sympk(\Vx_1'\MD_1\Vx_1) 
=\gamma_1 \wedge (\sympk(\Vxtilde'\MD_{2+}\Vxtilde)=\gamma_2)\right]=\\
&\frac{\fcpk(\MD_1,\gamma_1)\cdot\fcpk(\MD_{2+},\gamma_2)
}{\fcpk(\MD,t)}
\end{align*}
Then, we need to sample a pair $(a,b) \in \fspkt(\sym_1,\sym_2)$;
sample a uniform non-primitive $p^k$-representation of $a$ by $\MD_1$,
recursively sample a uniform non-primitive $p^k$-representation of $b$
by $\MD_{2+}$ and continue as in the non-primitive sampling case.

For sampling a uniform representation, after block diagonalization
we compute $\fbpk(\MD,t)$ and $\fcpk(\MD,t)$. We then sample a uniform
primitive representation with probability $\fbpk(\MD,t)/\fapk(\MD,t)$
and a uniform non-primitive representation with probability 
$\fcpk(\MD,t)/\fapk(\MD,t)$.

The computation of the number of ring operations (i.e., the ring
operation complexity) remains the same.
\end{proof}

\subsection{Sampling modulo a Composite Integer}

It is not very difficult to extend Theorem \ref{thm:SampleP} to
a composite integer modulus $q$ using the Chinese Remainder Theorem.
This argument is standard and was already used
in Siegel (Lemma~15, pp~544, \cite{Siegel35}). 

\begin{theorem}
Let $\MQ^n$ be an integral quadratic form, $t$ be an integer and
$q$ be a positive integer whose factorization is known. Then, there
is a $\poly(n, \log q, \log t)$ algorithm that counts (also, samples)
the solutions of $\Vx'\MQ\Vx \equiv t \bmod{q}$.
\end{theorem}
\begin{proof}
Given a quadratic form $\MQ^n$, and positive integer $t,q$ the task is to 
construct a uniform random solution of $\Vx'\MQ\Vx\equiv t \bmod q$. Suppose
that the factorization of $q$ is provided and $q=p_1^{k_1}\cdots p_r^{k_r}$.
Then, we proceed as follows. For each $i \in [r]$, sample a uniform
random solution $\Vx_i$ such that ${\Vx_i}'\MQ\Vx_i \equiv t \pmod{p_i^{k_i}}$.
Note that each $\Vx_i$ is an $n$-dimensional vector i.e., 
$\Vx_i \in (\bbZ/p_i^{k_i}\bbZ)^n$. Now, we solve the set of congruences
given below using the Chinese Remainder Theorem.
\begin{align*}
\Vx &\equiv \Vx_1 \bmod{p_1^{k_1}} \\
&\vdots \\
\Vx &\equiv \Vx_r \bmod{p_r^{k_r}} 
\end{align*}
By construction, $\Vx'\MQ\Vx \equiv t \bmod{q}$ and $\Vx$ is a uniform random
solution. Also,
\[
\fA_q(\MD,t) = \prod_{i \in [r]} \fA_{p_i^{k_i}}(\MD,t)\;.
\]

Given the factorization of $q$ this algorithm runs in polynomial
time i.e., $\poly(n,\log q, \log t)$. 
\end{proof}

\bibliographystyle{alpha}
\bibliography{quadraticforms}

\begin{thebibliography}{AEM87}

\bibitem[AEM87]{AEM87}
Leonard~M Adleman, Dennis~R Estes, and Kevin~S McCurley.
\newblock Solving bivariate quadratic congruences in random polynomial time.
\newblock {\em Mathematics of Computation}, 48(177):17--28, 1987.

\bibitem[Ajt96]{Ajtai96}
Mikl{\'o}s Ajtai.
\newblock Generating hard instances of lattice problems.
\newblock In {\em Proceedings of the twenty-eighth annual ACM symposium on
  Theory of computing}, pages 99--108. ACM, 1996.

\bibitem[Ank52]{Ankeny52}
NC~Ankeny.
\newblock The least quadratic non residue.
\newblock {\em Annals of mathematics}, pages 65--72, 1952.

\bibitem[Bac96]{BS96}
Eric Bach.
\newblock {\em Algorithmic Number Theory: Efficient Algorithms}, volume~1.
\newblock MIT press, 1996.

\bibitem[BS86]{BS86}
Zenon~Ivanovich Borevich and Igor~Rostislavovich Shafarevich.
\newblock {\em Number theory}, volume~20.
\newblock Academic Press, 1986.

\bibitem[CS99]{CS99}
John Conway and Neil~JA Sloane.
\newblock {\em Sphere packings, lattices and groups}, volume 290.
\newblock Springer, 1999.

\bibitem[Die03]{Dietmann03}
Rainer Dietmann.
\newblock Small solutions of quadratic diophantine equations.
\newblock {\em Proceedings of the London Mathematical Society},
  86(03):545--582, 2003.

\bibitem[Dix81]{Dixon81}
John~D Dixon.
\newblock Asymptotically fast factorization of integers.
\newblock {\em Mathematics of computation}, 36(153):255--260, 1981.

\bibitem[GY00]{GY00}
Wee~Teck Gan and Jiu-Kang Yu.
\newblock Group schemes and local densities.
\newblock {\em Duke Mathematical Journal}, 105(3):497--524, 2000.

\bibitem[Han04]{Hanke04}
Jonathan Hanke.
\newblock Local densities and explicit bounds for representability by a
  quadratic form.
\newblock {\em Duke Mathematical Journal}, 124(2):351--388, 2004.

\bibitem[Har08]{Hartung08}
Rupert Hartung.
\newblock {\em Computational problems of quadratic forms: complexity and
  cryptographic perspectives}.
\newblock PhD thesis, Ph. D. thesis, Goethe-Universit{\"a}t Frankfurt a. M.,
  2008, http://publikationen. ub. uni-frankfurt.
  de/volltexte/2008/5444/pdf/HartungRupert. pdf, 2008.

\bibitem[HR14]{HR13}
Ishay Haviv and Oded Regev.
\newblock On the lattice isomorphism problem.
\newblock {\em SODA}, pages 391--404, 2014.

\bibitem[IK04]{Kowalski04}
Henryk Iwaniec and Emmanuel Kowalski.
\newblock {\em Analytic number theory}, volume~53.
\newblock American Mathematical Society Providence, 2004.

\bibitem[Kit99]{Kitaoka99}
Yoshiyuki Kitaoka.
\newblock {\em Arithmetic of quadratic forms}, volume 106.
\newblock Cambridge University Press, 1999.

\bibitem[LLL82]{LLL82}
Arjen~Klaas Lenstra, Hendrik~Willem Lenstra, and L{\'a}szl{\'o} Lov{\'a}sz.
\newblock Factoring polynomials with rational coefficients.
\newblock {\em Mathematische Annalen}, 261(4):515--534, 1982.

\bibitem[Min10]{Minkowski10}
Hermann Minkowski.
\newblock {\em Geometrie der zahlen}.
\newblock Berlin, 1910.

\bibitem[O'M73]{OMeara73}
Onorato~Timothy O'Meara.
\newblock {\em Introduction to quadratic forms}, volume 117.
\newblock Springer, 1973.

\bibitem[Pal65]{Pall65}
Gordon Pall.
\newblock The weight of a genus of positive n-ary quadratic forms.
\newblock In {\em Proc. Sympos. Pure Math}, volume~8, pages 95--105, 1965.

\bibitem[Per52]{Per52}
Oskar Perron.
\newblock Bemerkungen {\"u}ber die verteilung der quadratischen reste.
\newblock {\em Mathematische Zeitschrift}, 56(2):122--130, 1952.

\bibitem[PS87]{PS87}
J~Pollard and C~Schnorr.
\newblock An efficient solution of the congruence.
\newblock {\em Information Theory, IEEE Transactions on}, 33(5):702--709, 1987.

\bibitem[Sho09]{Shoup09}
Victor Shoup.
\newblock {\em A computational introduction to number theory and algebra}.
\newblock Cambridge University Press, 2009.

\bibitem[Sie35]{Siegel35}
Carl~Ludwig Siegel.
\newblock {\"U}ber die analytische theorie der quadratischen formen.
\newblock {\em The Annals of Mathematics}, 36(3):527--606, 1935.

\bibitem[Sie72]{Siegel72}
Carl~Ludwig Siegel.
\newblock {\em Zur theorie der quadratischen formen}.
\newblock Vandenhoeck und Ruprecht, 1972.

\bibitem[Wat76]{Watson76}
GL~Watson.
\newblock The 2-adic density of a quadratic form.
\newblock {\em Mathematika}, 23(01):94--106, 1976.

\bibitem[Wed01]{Wedeniwski01}
Sebastian Wedeniwski.
\newblock {\em Primality Tests on Commutator Curves}.
\newblock PhD thesis, Eberhard-Karls-Universit\"at T\"ubingen, 2001.

\bibitem[Yan98]{Yang98}
Tonghai Yang.
\newblock An explicit formula for local densities of quadratic forms.
\newblock {\em Journal of Number Theory}, 72(2):309--356, 1998.

\end{thebibliography}

\appendix

\section{Diagonalizing a Matrix}\label{sec:BlockDiagonal}

In this section, we provide a proof of Theorem \ref{thm:BlockDiagonal}.

\paragraph{Module.} There are quadratic forms which have no associated 
lattice e.g., negative
definite quadratic forms. To work with these, we define the concept of
free modules (henceforth, called module) which behave as vector 
space but have no associated realization
over the Euclidean space $\bbR^n$.

If $M$ is finitely generated $\Ring$-module with generating set
$\Vx_1,\cdots,\Vx_n$ then the elements $\Vx \in M$ can
be represented as $\sum_{i=1}^n r_i \Vx_i$, such that
$r_i \in \Ring$ for every $i \in [n]$. By construction,
for all
$a,b \in R$, and $\Vx,\Vy \in M$;
\[
a(\Vx+\Vy)=a\Vx+a\Vy \qquad (a+b)\Vx=a\Vx+b\Vx \qquad a(b\Vx)=(ab)\Vx
\qquad 1\Vx=\Vx
\]
Note that, if we replace $\Ring$ by a field in the definition 
then we get a vector space (instead of a module). 
Any inner product
$\beta:M\times M \to \Ring$ gives rise to a quadratic form 
$\MQ\in\Ring^{n\times n}$ as follows;
\[
\MQ_{ij} = \beta(\Vx_i,\Vx_j) \;.
\]
Conversely, if $R=\bbZ$ then by definition, every symmetric matrix 
$\MQ \in \bbZ^{n\times n}$ gives rise to an inner product $\beta$ 
over every $\bbZ$-module $M$; as follows.
Given $n$-ary integral quadratic form $\MQ$ and a $\bbZ$-module
$M$ generated by the basis $\{\Vx_1,\cdots,\Vx_n\}$ we define the 
corresponding inner product $\beta:M\times M \to \bbZ$ as;
\[
\beta(\Vx,\Vy)=\sum_{i,j}c_id_j\MQ_{ij}
\text{ where, }\Vx=\sum_{i}c_i\Vx_i ~~ \Vy=\sum_{j}d_j\Vx_j\;.
\]
In particular, any integral quadratic form $\MQ^n$ can be interpreted 
as describing an inner product over a free module of dimension $n$.

For studying quadratic forms over $\zpkz$, where $p$ is a prime and $k$
is a positive integer; the first step is to find equivalent quadratic 
forms which have as few mixed terms as possible (mixed terms are terms
like $x_1x_2$).

\begin{proof}(Theorem \ref{thm:BlockDiagonal})
The transformation of the matrix $\MQ$ to a block diagonal form involves
three different kinds of transformation. We first describe these 
transformations on $\MQ$ with small dimensions (2 and~3).

\begin{enumerate}[(1)]
\item Let $\MQ$ be a $2\times 2$ integral quadratic form. Let us also 
assume that
the entry with smallest $p$-order in $\MQ$ is a diagonal entry, say
$\MQ_{11}$. Then, $\MQ$ is of the following form; where $\alpha_1,\alpha_2$
and $\alpha_3$ are units of $\zpz$.
\[
\MQ = \begin{pmatrix}p^i \alpha_1 & p^j \alpha_2 \\ p^j \alpha_2 & p^s\alpha_3 
\end{pmatrix} \qquad i \leq j,s
\]
The corresponding $\MU \in \text{SL}_2(\zpkz)$, that diagonalizes $\MQ$ 
is given below. The number
$\alpha_1$ is a unit of $\zpz$ and so $\alpha_1$ has an inverse in $\zpkz$.  
\[
\MU = \begin{pmatrix} 1 & -\frac{p^{j-i}\alpha_2}{\alpha_1} \bmod{p^k} \\ 
0 & 1\end{pmatrix} 
\qquad
\MU'\MQ\MU \equiv \begin{pmatrix} p^i\alpha_1 & 0 \\ 0 & p^s\alpha_3 - 
p^{2j-i}\frac{\alpha_2^2}{\alpha_1}\end{pmatrix}\pmod{p^k}
\]

\item If $\MQ^2$ does not satisfy the condition of item (1) i.e., 
the off diagonal entry is the one with smallest $p$-order, then we start
by the following transformation $\MV \in \SL_2(\zpkz)$.
\begin{gather*}
\MV = \begin{pmatrix}1 & 0 \\ 1 & 1\end{pmatrix} \qquad 
\MV'\MQ\MV = \begin{pmatrix} \MQ_{11}+2\MQ_{12}+\MQ_{22} & \MQ_{12}+\MQ_{22} \\
\MQ_{12}+\MQ_{22} & \MQ_{22}\end{pmatrix} 
\end{gather*}
If $p$ is an odd prime then $\ordp(\MQ_{11}+2\MQ_{12}+\MQ_{22})=\ordp(\MQ_{12})$, 
because $\ordp(\MQ_{11}),$ $\ordp(\MQ_{22}) >\ordp(\MQ_{12})$. By definition,
$\MS=\MV'\MQ\MV$ is equivalent to $\MQ$ over the ring $\zpkz$. But now, $\MS$
has the property that $\ordp(\MS_{11}) = \ordp(\MS_{12})$, and it can be 
diagonalized using the transformation in (1). The final transformation
in this case is the product of $\MV$ and the subsequent transformation
from item (1). The product of two matrices from $\SL_2(\zpkz)$ is also
in $\SL_2(\zpkz)$, completing the diagonalization in this case.

\item If $p=2$, then the transformation in item (2) fails. In this case,
it is possible to subtract a linear combination of these two rows/columns
to make everything else on the same row/column equal to zero over $\ztkz$.
The simplest such transformation is in dimension~3. The situation is as 
follows. Let $\MQ^3$ be a quadratic form whose off diagonal entry has the 
lowest possible power of
$2$, say $2^{\ell}$ and all diagonal entries are divisible by at least
$2^{\ell+1}$. In this case, the matrix $\MQ$ is of the
following form.
\[
\MQ = \begin{pmatrix}2^{\ell + 1} a & 2^{\ell}b & 2^id \\
2^{\ell}b & 2^{\ell+1}c & 2^je \\
2^id & 2^je & 2^{\ell+1}f \end{pmatrix} \qquad b \text{ odd}, \ell\leq i,j
\]
In such a situation, we consider the matrix $\MU \in \SL_3(\ztkz)$
of the form below such that if $\MS=\MU'\MQ\MU \pmod{2^k}$ then 
$\MS_{13}=\MS_{23}=0$.
\begin{gather*}
\MU = \begin{pmatrix}1 & 0 & -r \\ 0 & 1 & -s \\ 0 & 0 & 1\end{pmatrix}\\
(\MU'\MQ\MU)_{13} \equiv 0 \pmod{2^k} \implies 
r 2a + s b \equiv  2^{i-\ell}d \pmod{2^{k-\ell}}\\
(\MU'\MQ\MU)_{23} \equiv 0 \modtk \implies 
r b + s 2c \equiv 2^{j-\ell}e  \pmod{2^{k-\ell}}
\end{gather*}
For $i,j \geq \ell$ and $b$ odd, the solution $r$ and $s$ can be found by 
the Cramer's rule, as below. The solutions exist because the matrix 
$\begin{pmatrix}2a & b \\ b & 2c\end{pmatrix}$ has determinant $4ac-b^2$, which
is odd and hence invertible over the ring $\bbZ/2^{k-\ell}\bbZ$.
\[
r = \frac{\det\begin{pmatrix}2^{i-\ell}d & s \\ 2^{j-\ell}e & 2c\end{pmatrix}}
{\det\begin{pmatrix}2a & b \\ b & 2c\end{pmatrix}} \pmod{2^{k-\ell}} ~~
s = \frac{\det\begin{pmatrix}2a & 2^{i-\ell}d \\ b & 2^{j-\ell}e\end{pmatrix}}
{\det\begin{pmatrix}2a & b \\ b & 2c\end{pmatrix}}  \pmod{2^{k-\ell}}
\]
\end{enumerate}

This completes the description of all the transformations we are going
to use, albeit for $n$-dimensional $\MQ$ they will be a bit technical.
The full proof for the case of odd prime follows.

Our proof will be a reduction of the problem of diagonalization from
$n$ dimensions to $(n-1)$-dimensions, for the odd primes $p$. We now
describe the reduction.

Given the matrix $\MQ^n$, let $M$ be the corresponding $(\zpkz)$-module
with basis $\MB=[\Vb_1,\cdots,\Vb_n]$ i.e., $\MQ=\MB'\MB$. We first find 
a matrix entry with the smallest $p$-order, say $\MQ_{i^*j^*}$. The 
reduction has two cases: (i) there is a diagonal entry in $\MQ$ with
the smallest $p$-order, and (ii) the smallest $p$-order occurs on an
off-diagonal entry.

We handle case (i) first. Suppose it is possible to pick $\MQ_{ii}$
as the entry with the smallest $p$-order. Our first transformation
$\MU_1 \in \sln(\zpkz)$ is the one which makes the following 
transformation i.e., swaps $\Vb_1$ and $\Vb_i$.
\begin{align}\label{BlockDiagonal:U1}
[\Vb_1,\cdots,\Vb_n] \underset{\MU_1, p^k}{\to} [\Vb_i,\Vb_2,\cdots,
\Vb_{i-1},\Vb_1,\Vb_{i+1},\cdots, \Vb_n]
\end{align}

Let us call the new set of elements $\MB_1=[\Vv_1,\cdots,\Vv_n]$ and
the new quadratic form $\MQ_1=\MB_1'\MB_1 \bmod{p^k}$. Then,
$\Vv_1'\Vv_1$ has the smallest $p$-order in $\MQ_1$ and
$\MU_1'\MQ\MU_1\equiv \MQ_1 \bmod{p^k}$. The next transformation
$\MU_2 \in \sln(\zpkz)$ is as follows. 
\begin{equation}\label{BlockDiagonal:U2}
\Vw_i = \left\{ \begin{array}{ll}
\Vv_1 & \text{if $i=1$}\\
\Vv_i - \frac{\Vv_1'\Vv_i}{p^{\ordp((\MQ_1)_{11})}} \cdot 
\left(\frac{1}{\copp((\MQ_1)_{11})} \bmod{
p^k}\right) \cdot \Vv_1 & \text{otherwise\,.}
\end{array}\right.
\end{equation}
By assumption, $(\MQ_1)_{11}$ is the matrix entry with 
the smallest $p$-order and so $p^{\ordp((\MQ_1)_{11})}$ divides 
$\Vv_{1}'\Vv_i$. Furthermore, $\copp((\MQ_1)_{11})$ is invertible 
modulo $p^k$. Thus, the transformation in Equation 
\ref{BlockDiagonal:U2} is well defined. Also note that it is
a basis transformation, which maps one basis of $\MB_1=[\Vv_1,\cdots,\Vv_n]$
to another basis $\MB_2=[\Vw_1,\cdots,\Vw_n]$. Thus, the 
corresponding basis transformation $\MU_2$ is a 
unimodular matrix over integers, and so $\MU_2\in\sln(\zpkz)$.
Let $\MQ_2=\MU_2'\MQ_1\MU_2 \bmod{p^k}$. Then, we show that the
non-diagonal entries in the entire first row and first column of
$\MQ_2$ are~0. 
\begin{align*}
(\MQ_2&)_{1i(\neq 1)}=(\MQ_2)_{{i1}}=\Vw_1'\Vw_i \bmod{p^k}\\
	&\overset{(\ref{BlockDiagonal:U2})}{\equiv} 
	\Vv_1'\Vv_i - \frac{\Vv_1'\Vv_i}{p^{\ordp((\MQ_1)_{11})}} \cdot 
	\left(\frac{1}{\copp((\MQ_1)_{11})} \bmod{p^k}\right) 
	\cdot \Vv_1'\Vv_1 \\
	&\equiv\Vv_1'\Vv_i - \frac{\Vv_1'\Vv_i}{p^{\ordp((\MQ_1)_{11})}} \cdot 
	\left(\frac{1}{\copp((\MQ_1)_{11})} \bmod{p^k}\right) 
	\cdot p^{\ordp((\MQ_1)_{11})}\copp((\MQ_1)_{11}) \\
	&\equiv 0 \bmod{p^k} 
\end{align*}

Thus, we have reduced the problem to $(n-1)$-dimensions.
We now recursively call this algorithm with the quadratic form
$\MS=[\Vw_2,\cdots,\Vw_{n}]'[\Vw_2,\cdots,\Vw_{n}] \bmod{p^k}$
and let $\MV \in \SL_{n-1}(\zpkz)$ be the output of
the recursion. Then, $\MV'\MS\MV \bmod{p^k}$ is a diagonal
matrix. Also, by consruction $\MQ_2=\diag((\MQ_2)_{11},\MS)$.
Let $\MU_3=1\oplus\MV$, and $\MU=\MU_1\MU_2\MU_3$, then,
by construction, $\MU'\MQ\MU \bmod{p^k}$ is a diagonal
matrix; as follows.
\begin{gather*}
\MU'\MQ\MU \equiv \MU_3'\MU_2'\MU_1'\MQ\MU_1\MU_2\MU_3\equiv
\MU_3'\MQ_2\MU_3\equiv(1\oplus\MV)'\diag((\MQ_2)_{11})(1\oplus\MV)\\
\equiv \diag((\MQ_2)_{11},\MV'\MS\MV) \bmod{p^k}
\end{gather*}

Otherwise, we are in case (ii) i.e., the entry with smallest 
$p$-order in $\MQ$ is an off diagonal entry, say $\MQ_{i^*j^*},
i^*\neq j^*$. Then, we make the following basis transformation
from $[\Vb_1,\cdots,\Vb_n]$ to $[\Vv_1,\cdots,\Vv_n]$ as follows.
\begin{equation}\label{BlockDiagonal:U0}
\Vv_i = \left\{ \begin{array}{ll}
\Vb_{i^*}+\Vb_{j^*} & \text{if $i=i^*$}\\
\Vb_i & \text{otherwise\,.}
\end{array}\right.
\end{equation}
The transformation matrix $\MU_0$ is from $\sln(\zpkz)$.
Recall, 
$\ordp(\MQ_{i^*j^*}) < \ordp(\MQ_{i^*i^*}), \ordp(\MQ_{j^*j^*})$, and so
$\ordp(\Vv_{i^*}'\Vv_{i^*})=\ordp(\Vb_{i^*}'\Vb_{j^*})$. Furthermore, 
$\ordp(\Vv_i'\Vv_j)\geq \ordp(\Vb_{i^*}'\Vb_{j^*})$, and so the minimum 
$p$-order does not change after the transformation in Equation 
(\ref{BlockDiagonal:U0}). This transformation reduces the problem to the 
case when the matrix entry with minimum $p$-order appears on the 
diagonal. This completes the proof of the theorem for odd primes
$p$.

For $p=2$, exactly the same set of transformations works, unless the 
situation in item (3) arises. In such a case, we use the type II block
to eliminate all other entries on the same rows/columns as the type II
block. Thus, in this case, the problem reduces to one in dimension $(n-2)$.

The algorithm uses $n$ iterations, reducing the dimension by~1 in each 
iteration. In each iteration, we have to find the minimum $p$-order, costing
$O(n^2\log k)$ ring operations and then~3 matrix multiplications costing $O(n^3)$ operations
over $\zpkz$. Thus, the overall complexity is $O(n^4+n^3\log k)$ or 
$O(n^4\log k)$ ring operations.
\end{proof}

\section{Missing Proofs}\label{sec:Proofs}

\begin{proof}(proof of Lemma \ref{lem:Square})
We split the proof in two parts: for odd primes $p$ and for the prime~2.
\begin{description}
\item [Odd Prime.]
If $0 \neq t \in \zpkz$ then $\ordp(t)<k$. If
$t$ is a square modulo $p^k$ then there exists a $x$
such that $x^2\equiv t \pmod{p^k}$. Thus, there exists $a \in \bbZ$
such that $x^2=t+ap^k$. But then, $2\ordp(x)=\ordp(t+ap^k)=\ordp(t)$.
This implies that $\ordp(t)$ is even and $\ordp(x)=\ordp(t)/2$.
Substituting this into $x^2=t+ap^k$ and dividing the entire equation
by $p^{\ordp(t)}$ yields that $\copp(t)$ is a
quadratic residue modulo $p$; as follows.
\[
\copp(x)^2 = \copp(t)+ap^{k-\ordp(t)}\equiv \copp(t) \pmod{p}
\]

Conversely, if $\copp(t)$ is a quadratic residue modulo $p$ 
then there exists a $u \in \zpkz$ such that $u^2\equiv \copp(t) \pmod{p^k}$,
by Lemma \ref{lem:HenselP}. If $\ordp(t)$ is even then 
$x=p^{\ordp(t)/2}u$ is a solution to the equation 
$x^2\equiv t \pmod{p^k}$.

\item [Prime $2$.] 
If $0 \neq t \in \ztkz$ then $\ordt(t)<k$. If $t$ is a square modulo 
$2^k$ then there exists an integer $x$ such that $x^2\equiv t \modtk$.
Thus, there exists an integer $a$ such that $x^2=t+a2^k$. But then,
$2\ordt(x)=\ordt(t+a2^k)=\ordt(t)$. This implies that $\ordt(t)$ is 
even and $\ordt(x)=\ordt(t)/2$. Substituting this into the equation
$x^2=t+a2^k$ and dividing the entire equation by $2^{\ordt(t)}$ 
yields,
\[
\copt(x)^2 = \copt(t) + a2^{k-\ordt(t)} \qquad \copt(t) < 2^{k-\ordt(t)}\;.
\]
But $\copt(x)$ is odd and hence
$\copt(x)^2 \equiv 1 \pmod{8}$. If $k-\ordt(t)>2$, then 
$\copt(t)\equiv 1 \pmod{8}$. Otherwise, if $k-\ordt(t)\leq 2$ 
then $\copt(t) < 2^{k-\ordp(t)}$ implies that $\copt(t)=1$.

Conversely, if $\copt(t)\equiv 1 \pmod{8}$ then there exists a 
$u \in \ztkz$ such that $u^2\equiv \copt(t) \modtk$, by Lemma
\ref{lem:Hensel2}. If $\ordt(t)$ is even then $x=2^{\ordt(t)/2}u$
is a solution to the equation $x^2\equiv t\modtk$.
\end{description}
\end{proof}

\begin{proof}(proof of Lemma \ref{lem:SizeModP})
If $[+,-]$ denotes the size of the set of tuples $(x,x+a)$ such that 
$\legendre{x}{p}=1$ and $\legendre{x+a}{p}=-1$. Then, the following 
equality yields one relation on these sets.
\begin{align}
\sum_{x=1}^{p-1} \legendre{x(x+a)}{p} = \sum_{x=1}^{p-1} \legendre{x^2(1+ax^{-1})}{p}
&=\sum_{y=1}^{p-1}\legendre{1+y}{p}=\sum_{y=1}^{p-1} \legendre{y}{p}-\legendre{1}{p}
=-1\nonumber\\
[+,+]+[-,-]-[+,-]-[-,+] &= -1 \label{eq1:ordaeqordt}
\end{align}
If we consider all tuples $(x,x+a)$ for $x \in \{0,\cdots,p-1\}$ i.e., 
$(0,a), (1,a+1), \cdots, (p-1,a-1)$ then we observe that there are two 
tuples $(0,a)$ and $(p-a,0)$ with one Legendre symbol~0. We denote these
sets by $[0,+], [0,-], [+,0], [-,0]$. The size of these sets 
only depend on $\legendre{a}{p}$ and $\legendre{p-a}{p}=\legendre{-a}{p}$.
\begin{align}\label{eq2:ordaeqordt}
[0,+] &= \frac{1+\legendre{a}{p}}{2} \\
[0,-] &= \frac{1-\legendre{a}{p}}{2} \\
[+,0] &= \frac{1+\legendre{-a}{p}}{2} \\
[-,0] &=\frac{1-\legendre{-a}{p}}{2}
\end{align}
There are exactly $p$ tuples in total and exactly $2$ of them have 
one symbol $0$. This gives us the following relation.
\begin{align}\label{eq3:ordaeqordt}
[+,+]+[+,-]+[-,+]+[-,-] &= p-2 
\end{align}
Furthermore, if $p$ is an odd prime then there are exactly $\frac{p-1}{2}$
elements with Legendre symbol $+1$.
\begin{align}\label{eq4:ordaeqordt}
[+,+]+[+,-]+[+,0] &= \frac{p-1}{2} 
\end{align}
The bijection $(x,x+a)\to (-(x+a),-x)$ modulo $p$ maps the set $[+,+]$ to 
$[-,-]$ if $\legendre{-1}{p}=-1$ and the set $[+,-]$ to $[-,+]$ if 
$\legendre{-1}{p}=1$. This gives us the following relation.
\begin{align}\label{eq5:ordaeqordt}
[+,+] &= [-,-] &\textrm{if $\legendre{-1}{p}=-1$} \\
[+,-] &= [-,+] &\textrm{otherwise.} \label{eq6:ordaeqordt}
\end{align}
In each case i.e., $\legendre{-1}{p}=-1$ or $+1$, we have four Equations
i.e., (\ref{eq1:ordaeqordt}, \ref{eq3:ordaeqordt}, \ref{eq4:ordaeqordt}
,\ref{eq5:ordaeqordt}) or 
(\ref{eq1:ordaeqordt}, \ref{eq3:ordaeqordt}, \ref{eq4:ordaeqordt}
,\ref{eq6:ordaeqordt}) in four variables. These when solved, result in 
the following table of values, from which Equation (\ref{eq:ordaeqordt})
can be derived.

\begin{minipage}{\linewidth}
\centering
\captionof{table}{Size of split sets modulo $p$}\label{table}
\begin{tabular}{| l | l | c | r |}
\hline
$\legendre{a}{p}$ & $[\legendre{x}{p}, \legendre{x+a}{p}]$ & 
$\legendre{-1}{p}=1$ & $\legendre{-1}{p}=-1$\\
\hline
$+1$ & $[+1,+1]$ & $(p-1)/4-1$ & $(p-3)/4$\\
& $[-1,-1]$ & $(p-1)/4$ & $(p-3)/4$\\
& $[+1,-1]$ & $(p-1)/4$ & $(p+1)/4$\\
& $[-1,+1]$ & $(p-1)/4$ & $(p-3)/4$\\
\hline
$-1$ & $[+1,+1]$ & $(p-1)/4$ & $(p-3)/4$\\
&$[-1,-1]$ & $(p-1)/4-1$ &  $(p-3)/4$\\
&$[+1,-1]$ & $(p-1)/4$ & $(p-3)/4$\\
&$[-1,+1]$ & $(p-1)/4$ & $(p+1)/4$\\
\hline
\end{tabular}\par
\bigskip
\end{minipage}
\end{proof}

\end{document}